\documentclass[journal]{IEEEtran}
\usepackage{stfloats}
\usepackage{makecell}
\usepackage{graphicx}
\usepackage[cmex10]{amsmath}
\usepackage{cases}
\usepackage[tight,footnotesize]{subfigure}
\usepackage{amsthm} 
\usepackage{cite}
\usepackage{citesort}
\usepackage{amssymb}
\allowdisplaybreaks[4]
\usepackage{algorithm}
\usepackage{algorithmic}
\usepackage{multirow}
\usepackage{amsmath}
\usepackage{xcolor}
\usepackage{CJK}
\usepackage{subeqnarray}
\usepackage{cases}
\usepackage{enumerate}

\usepackage{stfloats}

\newtheorem{lemma}{\hskip\parindent\bf{Lemma}}

\newtheorem{theorem}{\hskip\parindent\bf{Theorem}}
\newtheorem{remark}{\hskip\parindent\bf{Remark}}

\ifCLASSINFOpdf
\else
\fi

\begin{document}

\title{UAV-Assisted Relaying and Edge Computing: Scheduling and Trajectory Optimization}

\author{Xiaoyan Hu,~\IEEEmembership{Student Member,~IEEE},
Kai-Kit~Wong,~\IEEEmembership{Fellow,~IEEE},
Kun Yang,~\IEEEmembership{Senior Member,~IEEE}, \\
and Zhongbin Zheng\\

\thanks{X. Hu and K.-K. Wong are with the Department of Electronic and Electrical Engineering, University College London, London  WC1E 7JE, UK (Email: \{xiaoyan.hu.16, kai-kit.wong\}@ucl.ac.uk).}
\thanks{K. Yang is with the University of Essex, Colchester CO4 3SQ, UK, and University of Electronic Science and Technology of China, Chengdu 611731, China (E-mail: kunyang@essex.ac.uk).}
\thanks{Z. Zheng is with the East China Institute of Telecommunications, China Academy of Information and Communications Technology, Shanghai 200001, China (Email: ben@ecit.org.cn).}
}
\maketitle
\begin{abstract}
In this paper, we study an unmanned aerial vehicle (UAV)-assisted mobile edge computing (MEC) architecture, in which a UAV roaming around the area may serve as a computing server to help user equipment (UEs) compute their tasks or act as a relay for further offloading their computation tasks to the access point (AP). We aim to minimize the weighted sum energy consumption of the UAV and UEs subject to the task constraints, the information-causality constraints, the bandwidth allocation constraints and the UAV's trajectory constraints. The required optimization is nonconvex, and an alternating optimization algorithm is proposed to jointly optimize the computation resource scheduling, bandwidth allocation, and the UAV's trajectory in an iterative fashion. Numerical results demonstrate that significant performance gain is obtained over conventional methods. Also, the advantages of the proposed algorithm are more prominent when handling computation-intensive latency-critical tasks.
\end{abstract}
\begin{IEEEkeywords}
UAV, mobile edge computing, resource scheduling, bandwidth allocation, trajectory optimization.
\end{IEEEkeywords}

\IEEEpeerreviewmaketitle

\section{Introduction}\label{sec:Introduction}
\subsection{Motivation and Prior Works}\label{sec:PW}
With the popularization of Internet of things (IoT) and the increasingly complex mobile applications, such as virtual and augmented reality, online gaming, automatic driving, etc., the computing demands at user equipment (UEs) are reaching an unprecedented level.
Mobile edge computing (MEC), widely regarded as the technology to help the resource-limited UEs handle computing-intensive latency-critical tasks, has attracted great attention from both the academia and the industry. The standardization organizations and industry associations such as  ETSI and 5GAA have identified several use cases for MEC, from the  intelligent video acceleration and application-aware performance optimization to vehicle-to-everything and massive machine-type communications, etc. \cite{W_Y.Hu15Mobile,5GAA-VISUAL}.

The rationale behind MEC is that UEs' computing tasks can be offloaded and completed at the edge of wireless networks by deploying cloud servers at the access points (APs), so as to liberate the UEs from heavy computing workloads and prolong their battery lifetime~\cite{S_P.Mach17Mobile,S_Y.Mao17ASurve}. Recently, MEC has been widely used in cellular networks, focusing on improving the energy efficiency or reducing the latency of various cellular-based MEC systems \cite{S.Sardellitti.15Joint,J_C.You17Energy,T.Q.Dinh17Offloading,J_X.Hu18Edge,J_L.Pu16D2DFogging,J_H.Sun19Joint,J_C.You2016EnergyEM,J_Y.Mao16Dynamic,J_X.Hu18Wireless,J_F.Wang18Joint}. A  multicell MEC system was studied in \cite{S.Sardellitti.15Joint}, where the total energy consumption was minimized by jointly optimizing the radio and computational resources.
In \cite{J_C.You17Energy}, the resource allocation for minimizing the weighted sum energy consumption of users was addressed with a derived threshold-based optimal policy.
Later in \cite{T.Q.Dinh17Offloading}, the scenario of a UE with multiple tasks was considered, where multiple APs assisted the UE to reduce its total task execution latency and energy consumption.
A two-tier heterogeneous network with the coexistence of edge  and central cloud computing was studied in \cite{J_X.Hu18Edge},  and the cloud selection was optimized to minimize the network's energy consumption.
In \cite{J_L.Pu16D2DFogging}, a device-to-device (D2D) fogging was explored to achieve energy-efficient task completion by sharing computation and communication resources amongst mobile devices. The sum of computation efficiency defined as the calculated data bits divided by the energy consumption was maximized in \cite{J_H.Sun19Joint} with iterative and gradient descent methods.
In addition, the works in \cite{J_C.You2016EnergyEM,J_Y.Mao16Dynamic,J_X.Hu18Wireless,J_F.Wang18Joint} introduced the use of energy harvesting or wireless power transfer (WPT) technologies into the cellular-based MEC systems, which has enabled the UEs to have sustainable energy support to their transmissions and computation, but at the cost of increasing the computational complexity of the systems.

Due to the attractive advantages of unmanned aerial vehicle (UAV) for its easy deployment, flexible movement, and line-of-sight (LoS) connections, and so on, UAV-enabled wireless communication networks have been much researched in recent years \cite{J_zeng16Wireless,J_Zeng17Energy,J_Y.Zeng16Throughput,J_M.M.Azari18Ultra,J_Xu18Uav}.
For instance, an energy-efficient UAV communication was investigated in \cite{J_Zeng17Energy}, in which an UAV flew at a fixed altitude and had the  initial and final locations preset on its trajectory design.
In \cite{J_Y.Zeng16Throughput}, the UAV-enabled mobile relaying systems were studied, where the throughput was maximized by optimizing the transmit power allocation and the UAV's trajectory.
Recently, \cite{J_M.M.Azari18Ultra} proposed a generic framework for the analysis and optimization of the air-to-ground systems, and an optimum altitude for UAV in maximizing the coverage region with a guaranteed minimum outage performance was derived. WPT technology was considered for UAV wireless networks in \cite{J_Xu18Uav}, and the UAV trajectory was optimized to maximize the sum energy or the minimum energy transferred to all the UEs. It was revealed that UAV-enabled WPT can significantly enhance the WPT performance over the traditional WPT system with fixed energy transmitters.

It is a great attempt to leverage the technology of the UAV in  MEC systems, and the performance improvement of the UAV-enabled MEC architecture has been shown to be substantial \cite{J_Jeong18Mobile,J_F.Zhou18Computation,C_X.Cao2018MobileEC}. A UAV-based MEC system was investigated in \cite{J_Jeong18Mobile}, where a moving UAV equipped with a processing server was considered to help UEs compute their offloaded tasks. The total mobile energy consumption was minimized by jointly optimizing the task-bit allocation and the UAV trajectory  using the successive convex approximation (SCA) methods. Later in \cite{J_F.Zhou18Computation}, a wireless powered UAV-enabled MEC system was studied, where the UAV was endowed with an energy transmitter and an MEC server to provide energy as well as MEC services for the UEs. The computation rate maximization problems were addressed under both the partial and binary computation offloading modes by alternating algorithms. In another study \cite{C_X.Cao2018MobileEC}, the UAV acted as a UE rather than an MEC server, which was served by multiple cellular ground base stations  to compute its offloaded tasks. The UAV's  mission completion time was minimized by optimizing the resource allocation and the UAV trajectory through an SCA algorithm.

\subsection{Our Contributions}\label{sec:Contributions}
The aforementioned MEC works concentrate either on the cellular-based MEC networks, where the UEs' tasks are completed by using the computing resources at the APs; or the UAV-enabled MEC architectures by exploiting the computing capability of the UAV processing server. However, for the UEs with seriously degraded links to the AP due to severe blockage, it is impossible to take full use of the computing resources at the AP directly. Besides, due to the size-constrained resource-limited property of the UAVs, it is risky to rely only on the UAVs to assist the UEs for completing their computation-intensive latency-critical tasks. For these reasons, this paper studies a UAV-assisted MEC architecture, where the computing resources at the UAV and the AP are utilized at the same time. In addition, the energy-efficient LoS transmissions of the UAV have been fully exploited since the UAV is not only served as a mobile computing server to help the UEs compute their tasks  but also as a relay to further offload UEs' tasks to the AP for computing. To our best knowledge, this is the first work considering the UAV-assisted  MEC architecture by letting the UAV act as an MEC server and a relay simultaneously.

Our main contributions are summarized as follows:
\begin{itemize}
  \item \textbf{UAV-Assisted  MEC Architecture}---We consider a UAV-assisted  MEC architecture where the cellular-connected UAV is served as a mobile computing server as well as a relay to help the UEs complete their computing tasks or further offload their tasks to the AP for computing. This architecture takes full advantages of the UAV's energy-efficient LoS transmissions, and makes proper use of the computing resources at both the UAV and AP.

  \item \textbf{Problem Formulation with Joint Computation Resource Scheduling, Bandwidth Allocation and UAV's Trajectory Optimization}---Our aim is to minimize the weighted sum energy consumption (WSEC) of the UAV and the UEs subject to the UEs' task constraints, the information-causality constraints, the bandwidth allocation constraints and the UAV's trajectory constraints, by jointly optimizing the computation resource scheduling, the bandwidth allocation and UAV's trajectory iteratively. The formulated problem is complicated and non-convex due to the coupled optimization variables.

  \item \textbf{Alternating Algorithm with Guaranteed Convergence} ---An alternating optimization algorithm is devised to decouple the optimization variables, through which the formulated problem can be properly solved by addressing three subproblems iteratively. Note that the computation resource scheduling parameters, including the offloading/downloading task sizes and the CPU frequencies at each UE and the UAV, as well as the bandwidth allocation parameters are obtained in closed form by leveraging the Lagrange duality method, and that the corresponding Lagrange multipliers associated with the inequality constraints can be obtained using the subgradient method while those associated with the equality constraints can be obtained through bi-section search. The subproblem relating to the UAV's trajectory optimization can be efficiently solved by CVX \cite{M_Grant08CVX} based on the SCA method. Besides, the convergence of the proposed algorithm can be guaranteed, and the required complexity appears to be acceptable.

  \item {\textbf{Considerable Performance Improvement}}---Simulation results are presented to show the optimized trajectories of the UAV under different scenarios and the significant performance enhancement by leveraging the proposed algorithm when compared to existing schemes, such as the one with a preset UAV trajectory, the scheme with task offloading only, the scheme with equal bandwidth allocation, and the local computing scheme without offloading. Moreover, the proposed algorithm is capable of providing more stable performance in adapting to the change in the operating environment, and its advantages will become much more prominent when dealing with the computation-intensive and latency-critical tasks.
\end{itemize}

The rest of this paper is organized as follows. In Section~\ref{sec:system}, we introduce our system model and then formulate the optimization problem. The proposed method that decouples the problem into three subproblems then solving it iteratively is presented in Section~\ref{algorithm_design1}. Section~\ref{sec:simulation} provides the simulation results. Finally, we conclude the paper in Section~\ref{sec:conclusion}.


\section{System Model and Problem Formulation}\label{sec:system}
\begin{figure}[t!]
\centering
\includegraphics[width=2.7in]{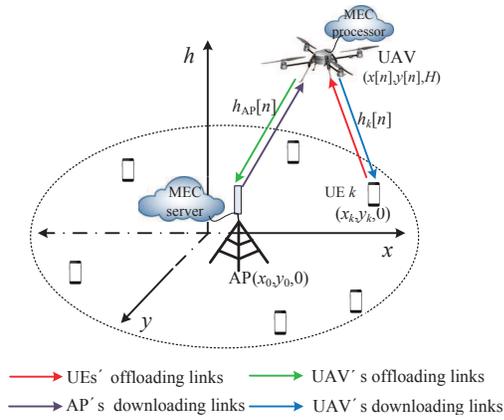}
\caption{An illustration of the UAV-assisted MEC architecture, where the UAV serves as an MEC server to help the ground UEs compute their offloaded tasks or as a possible relay to further forward the offloaded tasks to the AP with more powerful computing resources.}
\label{UAV_relay_MEC}
\end{figure}

As shown in Fig.~\ref{UAV_relay_MEC}, a UAV-assisted  MEC system is considered, which consists of an AP, a cellular-connected UAV, and $K$ ground UEs, all being equipped with a single antenna. The UAV and UEs are all assumed to have an on-board communication circuit and on-board computing processor powered by their embedded battery, while the AP is capable of providing high-speed transmission rate with grid power supply and is endowed with an ultra-high performance processing server. It is also assumed that each UE has a bit-wise-independent computation-intensive task,
and the UAV acts as an assistant to  help the UEs complete their computation tasks by providing both  MEC and relaying services. For providing MEC service, the UAV shares its computing resources with the UEs to help compute their tasks; while for the relaying service, the UAV forwards part of the UEs' offloaded tasks to the AP for computing with the purpose of saving its own energy.

\subsection{Channel Model and Coordinate System}\label{sec:Channel Model}
A three-dimensional (3D) Euclidean coordinate system is adopted, whose coordinates are measured in meters. We assume that the locations of the AP and all the UEs are fixed on the ground with zero altitude, with the location of the AP being $\mathbf{\widetilde{v}}_0=(x_0,y_0,0)$. Let $\mathcal{K}=\{1,\dots,K\}$ denote the set of the UEs, with $\mathbf{\widetilde{v}}_k=(x_k,y_k,0)$ representing the location of UE $k\in\mathcal{K}$. It is assumed that the locations of UEs are known to the UAV for designing its trajectory \cite{J_Zeng17Energy}.
We assume that the UAV flies at a fixed altitude $H>0$ during the task completion time $T$, which corresponds to the minimum altitude that is appropriate to the work terrain and can avoid buildings without the requirement of frequent descending and ascending.

For ease of exposition, the finite task completion time $T$ is discretized into $N$ equal time slots each with a duration of $\tau=T/N$,
where $\tau$ is sufficiently small such that the UAV's location can be assumed to be unchanged during each slot.
The initial and final horizontal locations of the UAV are preset as $\mathbf{u}_\mathrm{I}=(x_\mathrm{I}, y_\mathrm{I})$ and $\mathbf{u}_\mathrm{F}=(x_\mathrm{F}, y_\mathrm{F})$, respectively. Let $\mathcal{N}=\{1,\dots,N\}$ denote the set of the $N$ time slots. At the $n$-th time slot, the UAV's horizontal location is denoted as $\mathbf{u}[n]\equiv\mathbf{u}(n\tau)=(x[n],y[n])$ with $\mathbf{u}[0]=\mathbf{u}_\mathrm{I}$ and $\mathbf{u}[N]=\mathbf{u}_\mathrm{F}$. It is assumed that the UAV flies with a constant speed in each time slot, denoted as $v[n]$, which should satisfy the following maximum speed constraint
\begin{align}\label{eq:speed_n}
v[n]=\frac{\|\mathbf{u}[n]-\mathbf{u}[n-1]\|}{\tau}\leq V_\mathrm{max}, \ n\in \mathcal{N},
\end{align}
where $V_\mathrm{max}$ is the predetermined maximum speed of the UAV, and $V_\mathrm{max}\geq\|\mathbf{u}_\mathrm{F}-\mathbf{u}_\mathrm{I}\|/T$ establishes to make sure that at least one feasible trajectory of the UAV exists.

Similar to \cite{J_Zeng17Energy},
the wireless channels between the UAV and the AP as well as the UEs are assumed to be dominated by LoS links, which is verified by recent field experiments done by Qualcomm \cite{R_Qualcomm17Unmanned}.\footnote{It is of great value to extend our work on the probabilistic LoS and Rician fading channel models when we consider the scenarios where the UAV's flying altitude changes according to the work terrain.} Thus, the channel power gain between the UAV and the AP and between the UAV and  UE $k$ at the time slot $n$ can be, respectively, given by
\begin{align}
&h_{\mathrm{AP}}[n]=h_0d_{\mathrm{AP}}^{-2}=\frac{h_0}{\|\mathbf{u}[n]-\mathbf{v}_0\|^2+H^2}, \ n\in \mathcal{N}, \label{eq:channel_AP} \\
&h_k[n]=h_0d_{k}^{-2}=\frac{h_0}{\|\mathbf{u}[n]-\mathbf{v}_k\|^2+H^2}, \ k\in\mathcal{K}, n\in \mathcal{N}, \label{eq:channel_k}
\end{align}
where $h_0$ is the channel power gain at a reference distance of $d_0=1$m; $d_{\mathrm{AP}}$ and  $d_{k}$ are respectively the horizontal plane distances between the UAV and the AP  as well as the UE $k$ at the $n$-th time slot with $\mathbf{v}_0=(x_0,y_0)$ and $\mathbf{v}_k=(x_k,y_k)$ denoting the horizontal locations of the AP and UE $k$, $k\in\mathcal{K}$.
It is assumed that the channel reciprocity establishes in our considered scenario, and thus the offloading and downloading channels between the UEs and the UAV are identical. In this paper, the direct links between UEs and the AP are assumed to be negligible due to e.g., severe blockage,\footnote{The general case with direct links between the UEs and the AP will be considered as one of our future works.} which means that the UEs cannot directly offload their task-input bits to the AP unless with the assistance of the UAV.
The motivation behind this scenario is based on the fact that it is more important to guarantee the UEs' computation tasks being completed within the given limited time $T$ with as little UEs' energy as possible, than dropping their tasks or letting the UEs compute their takes locally at the cost of exhausting their energy.

\subsection{Computation Task Model and Execution Methods}\label{sec:Task Model}
The computation task of UE $k\in\mathcal{K}$  is denoted as a positive tuple $[I_k, C_k, O_k, T_k]$, where $I_k$ denotes the size (in bits) of the computation task-input data (e.g., the program codes and input parameters), $C_k$ is the amount of required computing resource for computing 1-bit of input data (i.e., the number of CPU cycles required), $O_k\in(0,1)$ is the ratio of task-output data size to that of the task-input data, i.e., the output data size should be $O_kI_k$, and $T_k$ is the maximum tolerable latency with $T_k\leq T, k\in \mathcal{K}$. In this paper, we only consider the case that $T_k=T$ for all $k\in\mathcal{K}$. It should be noted that the UEs' task-input bits are bit-wise independent and can be arbitrarily divided to facilitate parallel trade-offs  between local computing at the UEs and computation offloading to the UAV or further to the AP with the assistance of the UAV. In other words, the UEs can accomplish their computation tasks in a partial offloading fashion \cite{S_Y.Mao17ASurve} with the following three ways. 
\subsubsection{Local Computing at UEs}
Each UE can perform local computing and computation offloading simultaneously since local computing at the UEs does not need radio resources such as bandwidth. To efficiently use the energy for local computing, the UEs leverage a dynamic voltage and frequency scaling (DVFS) technique, and thus the energy consumed for local computing can be adaptively controlled  by adjusting the UEs' CPU frequency during each time slot \cite{J_Zhang13Energy}. The CPU frequency of  UE $k$ during time slot $n$ is denoted as $f_k[n]$ (cycles/second). Thus, the computation bits and energy consumption of UE $k$ during time slot $n$ for local computing can be, respectively, expressed as\footnote{All the
energy consumption in this paper uses the unit of Joule.}
\begin{align}
&L{_k^{\mathrm{local}}}[n]=\tau f_k[n]/C_k, \ k\in\mathcal{K}, n\in \mathcal{N}, \label{eq:Local_size_k}\\
&E{_k^{\mathrm{local}}}[n]=\tau\kappa_k f_k^3[n], \ k\in\mathcal{K}, n\in \mathcal{N}, \label{eq:Local_energy_k}
\end{align}
where $\kappa_k$ is the effective capacitance coefficient of UE $k$ that depends on  its processor's chip architecture.

\subsubsection{Task Offloaded to the UAV for Computing}
The UEs' remaining task-input data should be computed remotely, first by offloading to the UAV, and then one part of the data being computed at the UAV while the other part further offloaded to the AP for computing. In order to avoid interference among the UEs during the offloading process, we adopt the time-division multiple access (TDMA) protocol.
Each slot is further divided into $K$ equal durations $\delta=T/(NK)$, and UE $k$ offloads its task-input data in the $k$-th duration. Let $l_k[n]$ denote the offloaded bits of UE $k$ in its allocated duration at time slot $n$, and thus  the corresponding energy consumption of UE $k$ at slot $n$ for computation offloading can be calculated  as
\begin{align}\label{eq:E_off_k}
E{_k^{\mathrm{off}}}[n]&=\delta p_k[n]    \nonumber\\
&\equiv\frac{\delta N_0}{h_k[n]}\left(2^{\frac{l_k[n]}{\delta B{_k^\mathrm{off}}[n]}}-1\right), \ k\in \mathcal{K}, n\in \mathcal{N},
\end{align}
where $p_k[n]$
is the transmit power of UE $k$ for offloading $l_k[n]$ computation bits to the UAV at time slot $n$, $B{_k^\mathrm{off}}[n]$ is the corresponding allocated bandwidth for UE $k$, and $N_0$ denotes the noise power at the UAV.\footnote{Without loss of generality, we assume that the noise power at any node in the system is considered the same as $N_0$.}

Assume that the UAV also adopts the DVFS technique to improve its energy efficiency for computing, and its adjustable CPU frequency in the $k$-th duration of slot $n$ for computing UE $k$'s offloaded task is denoted as $f_{\mathrm{U},k}[n]$. Hence, the completed computation bits and the energy consumption of the UAV for computing UE $k$'s task at slot $n$ can be, respectively, given by
\begin{align}
L_{\mathrm{U},k}[n]&=\delta f_{\mathrm{U},k}[n]/C_k, \ k\in\mathcal{K}, n\in \mathcal{N}, \label{eq:UAV_size_k}\\
E_{\mathrm{U},k}[n]&=\delta\kappa_{\mathrm{U}} f_{\mathrm{U},k}^3[n], \ k\in\mathcal{K}, n\in \mathcal{N}, \label{eq:UAV_energy_k}
\end{align}
where $\kappa_\mathrm{U}$ is the effective capacitance coefficient of the UAV. Note that computing $L_{\mathrm{U},k}[n]$ bits of UE $k$'s task-input data will produce $O_kL_{\mathrm{U},k}[n]$ bits of task-output data, which should be downloaded from the UAV to the UE $k$ later.

\subsubsection{Task Offloaded to the AP for Computing}
Part of the UEs' offloaded task-input data at the UAV will be offloaded to the AP's processing server for computing. To better distinguish the offloading signals from different UEs, the TDMA protocol with $K$ equal time division ($\delta=T/(NK)$) is also adopted in this case. Let $l{_{\mathrm{U},k}^{\mathrm{off}}}[n]$ denote the number of UE $k$'s task-input bits being offloaded from the UAV to the AP at time slot $n$. Thus, the corresponding energy consumption of the  UAV for offloading UE $k$'s task at slot $n$ can be calculated as
\begin{align}\label{eq:E_UAV_off_k}
E{_{\mathrm{U},k}^{\mathrm{off}}}[n]&=\delta p{_{\mathrm{U},k}^{\mathrm{off}}}[n]  \nonumber\\
&\equiv\frac{\delta N_0}{h_{\mathrm{AP}}[n]}\Bigg(2^{\frac{l{_{\mathrm{U},k}^{\mathrm{off}}}[n]}{\delta B{_{\mathrm{U},k}^\mathrm{off}}[n]}}-1\Bigg), \ k\in \mathcal{K}, n\in \mathcal{N},
\end{align}
where $p{_{\mathrm{U},k}^{\mathrm{off}}}[n]$
and $B{_{\mathrm{U},k}^\mathrm{off}}[n]$ are respectively the transmit power  and  the allocated bandwidth of the UAV for offloading UE $k$'s tasks to the AP at time slot $n$. After computing the $l{_{\mathrm{U},k}^{\mathrm{off}}}[n]$ input bits at the AP, $O_kl{_{\mathrm{U},k}^{\mathrm{off}}}[n]$  bits of computation results for UE $k$ will be generated. As the AP is integrated with an ultra-high-performance processing server, the computing time is negligible. The AP will send the computation results back to the UAV in the TDMA manner using a separate bandwidth. Since the AP is supplied with grid power and can support ultra-high transmission rate, the download transmission time from the AP to the UAV is also assumed negligible.\footnote{Once the AP receives the forwarded $l{_{\mathrm{U},k}^{\mathrm{off}}}[n]$ bits input data from the UAV in the $k$-th duration of the $n$-th time slot, it will immediately decode, compute the data, and then send the induced $O_kl{_{\mathrm{U},k}^{\mathrm{off}}}[n]$  bits of output data back to the UAV, all with ultra-low latency that is negligible compared with the length of each duration $\delta$, which means that the UAV can receive the task-output data from the AP in the same duration of its offloading process.}

For the later two offloading methods, the generated computation results at the UAV (including the results from UAV's computing and received from the AP) will then be downloaded back to the corresponding UEs. It is assumed that the UAV is equipped with a data buffer with sufficiently large size, and it is capable of storing each UE's offloaded data and the corresponding computation results separately. Besides, we assume that the UAV operates in a frequency-division-duplex (FDD) mode in each UE's operation duration $\delta$ with separate bandwidths allocated for task reception from UEs ($\{B{_k^\mathrm{off}}[n]\}$), task offloading transmission to the AP ($\{B{_{\mathrm{U},k}^{\mathrm{off}}}[n]\}$), and task results downloading transmission to the UEs (\{$B{_{\mathrm{U},k}^{\mathrm{down}}}[n]$\}), with a total bandwidth $B$ satisfying the constraint
\begin{equation}
B{_k^\mathrm{off}}[n]+B{_{\mathrm{U},k}^{\mathrm{off}}}[n]+B{_{\mathrm{U},k}^{\mathrm{down}}}[n]=B,\ k\in\mathcal{K}, n\in\mathcal{N}.
\end{equation}

The UEs' computation results are subsequently transmitted by the UAV using TDMA similar to the UEs' offloading process, each with an equal duration $\delta$ in each time slot. Let $l{_{\mathrm{U},k}^{\mathrm{down}}}[n]$ denote the bits of task-output data being downloaded from the UAV to UE $k$  at time slot $n$. Hence, the corresponding energy consumption of the UAV can be calculated as
\begin{align}\label{eq:E_UAV_down_k}
E{_{\mathrm{U},k}^{\mathrm{down}}}[n]&=\delta p{_{\mathrm{U},k}^{\mathrm{down}}}[n] \nonumber\\
&\equiv\frac{\delta N_0}{h_{k}[n]}\Bigg(2^{\frac{l{_{\mathrm{U},k}^{\mathrm{down}}}[n]}{\delta B{_{\mathrm{U},k}^\mathrm{down}}[n]}}-1\Bigg), \ k\in \mathcal{K}, n\in \mathcal{N},
\end{align}
where $p{_{\mathrm{U},k}^{\mathrm{down}}}[n]$
is the transmit power of the UAV for downloading UE $k$'s task-output data at time slot $n$.

Note that at each time slot $n$, the UAV can only compute or forward the task-input data that has already been received from the UEs. By assuming that the processing delay, e.g., the delay for decoding and computing preparation, at the UAV is one time slot, then we have the following information-causality constraint:
\begin{align}\label{eq:UE_UAV_Causality}
\sum\limits_{i{\rm{ = }}2}^n \left(\frac{\delta f_{\mathrm{U},k}[i]}{C_k}+l{_{\mathrm{U},k}^{\mathrm{off}}}[i]\right) \leq \sum\limits_{i{\rm{ = }}1}^{n-1} l{_{k}}[i],
\end{align}
for $ n\in\mathcal{N}_2$ and $k\in\mathcal{K}$ where $\mathcal{N}_2=\{2,\dots,N-1\}$. Similarly, at each time slot $n$, the UAV can only transmit the task-output data corresponding to the task-input data that has already been computed at the UAV or offloaded for computing at the AP. Thus, we have another information-causality constraint:
\begin{align}\label{eq:UAV_AP_Causality}
\sum\limits_{i{\rm{ = }}3}^{n} l{_{\mathrm{U},k}^{\mathrm{down}}}[i]\leq O_k\sum\limits_{i{\rm{ = }}2}^{n-1} \left(\frac{\delta f_{\mathrm{U},k}[i]}{C_k} +l{_{\mathrm{U},k}^{\mathrm{off}}}[i]\right),
\end{align}
for $ n\in\mathcal{N}_3$ and $k\in\mathcal{K}$ where $\mathcal{N}_3=\{3,\dots,N\}$.
It is clear that the UEs should not offload at the last two slots, while the UAV should not compute or forward the received input data of UEs' at the first and the last slots as well as  not transmit the output data to the UEs in the first two slots. Hence, we have  $l{_{k}}[N-1]=l{_{k}}[N]=0$, $f_{\mathrm{U},k}[1]=f_{\mathrm{U},k}[N]=0$, $l{_{\mathrm{U},k}^{\mathrm{off}}}[1]=l{_{\mathrm{U},k}^{\mathrm{off}}}[N]=0$, and $l{_{\mathrm{U},k}^{\mathrm{down}}}[1]=l{_{\mathrm{U},k}^{\mathrm{down}}}[2]=0$.

\subsection{Problem Formulation}\label{sec:problem}
Considering the fact that the traditional battery-based UEs and  UAVs are usually power-limited, one major problem the UAV-assisted MEC system faces will be energy. Hence, in this paper, we try to minimize the WSEC of the UAV as well as all the UEs during the whole task completion time $T$. In the previous subsection, we have obtained the energy consumption of the UEs and the UAV for task offloading/downloading and computation. In fact, the energy consumption for UAV's propulsion is also considerable which is greatly affected by the UAV's trajectory, and hence should be taken into account.
With the assumption that the time slot duration $\tau$ is  sufficiently small, the UAV's flying during each slot can be regarded as  straight-and-level flight with constant speed $v[n]$.
Taking a fixed-wing UAV as an example \cite{J_Zeng17Energy,J_Y.Zeng19Accessing}, its propulsion energy consumption
at time slot $n$ can be expressed as
\begin{align}\label{eq:E_fly}
E{_\mathrm{U}^{\mathrm{fly}}}[n]=\tau\bigg(\theta_1v^3[n]+\frac{\theta_2}{v[n]}\bigg), \ n\in \mathcal{N},
\end{align}
where $\theta_1$ and $\theta_2$ are two parameters related to the UAV's weight, wing area, wing span efficiency, and air density, etc.
Combining with the above analysis, we obtain the total energy consumption of UE $k$ and the UAV in each time slot $n$ as
\begin{align}
E_{k}[n]&=E{_k^{\mathrm{local}}}[n]+E{_k^{\mathrm{off}}}[n], \ k\in\mathcal{K}, n\in \mathcal{N},\label{eq:E_k_n}\\
E_{\mathrm{U}}[n]&=\sum\limits_{k{\rm{ = }}1}^K \Big(E_{\mathrm{U},k}[n]+E{_{\mathrm{U},k}^{\mathrm{off}}}[n]+\Big. \nonumber\\
&~~~~~~~~~~\Big.E{_{\mathrm{U},k}^{\mathrm{down}}}[n]\Big)+E{_\mathrm{U}^{\mathrm{fly}}}[n], \  n\in \mathcal{N}.\label{eq:E_UAV_k_N}
\end{align}

In our considered scenario, the UEs' CPU computing frequencies $\{f_k[n]\}$, their offloading task-input bits $\{l{_{k}}[n]\}$ and the corresponding allocated bandwidth $\{B{_k^\mathrm{off}}[n]\}$; the UAV's CPU computing frequencies $\{f_{\mathrm{U},k}[n]\}$, its forwarding (further offloading) task-input bits $\{l{_{\mathrm{U},k}^{\mathrm{off}}}[n]\}$ and downloading task-output bits $\{l{_{\mathrm{U},k}^{\mathrm{down}}}[n]\}$ as well as the corresponding allocated bandwidths $\{B{_{\mathrm{U},k}^{\mathrm{off}}}[n]\}$, $\{B{_{\mathrm{U},k}^{\mathrm{down}}}[n]\}$ for different UEs; along with the UAV's trajectory $\{\mathbf{u}[n]\}$ will be optimized to minimize the WSEC.
To this end, the WSEC minimization problem can be formulated as problem (P1) given below
\vspace{-3mm}

{\small{
\begin{subeqnarray}\label{eq:WSECM1}
&&\hspace{-8mm}\underset{\mathbf{z},\mathbf{B},\mathbf{u}}{\min}~~
\sum\limits_{n{\rm{ = }}1}^N\left(w_\mathrm{U}E_{\mathrm{U}}[n]+\sum\limits_{k{\rm{ = }}1}^Kw_kE_k[n]\right)\slabel{eq:WSECM1_0}\\
&&\hspace{-8mm} \mathrm{s.t.}~~
\sum\limits_{i{\rm{ = }}2}^n \left(\frac{\delta f_{\mathrm{U},k}[i]}{C_k}+l{_{\mathrm{U},k}^{\mathrm{off}}}[i]\right)\leq \sum\limits_{i{\rm{ = }}1}^{n-1} l{_{k}}[i],\ \forall n\in\mathcal{N}_2,\ \forall k\in\mathcal{K}, \slabel{eq:WSECM1_1}\\
&&\hspace{-8mm} \sum\limits_{i{\rm{ = }}3}^n l{_{\mathrm{U},k}^{\mathrm{down}}}[i]\leq O_k\sum\limits_{i{\rm{ = }}2}^{n-1} \left(\frac{\delta f_{\mathrm{U},k}[i]}{C_k} +l{_{\mathrm{U},k}^{\mathrm{off}}}[i]\right), \forall n\in\mathcal{N}_3, \forall k\in\mathcal{K}, \quad\quad\slabel{eq:WSECM1_2}\\
&&\hspace{-8mm} \sum\limits_{n{\rm{ = }}2}^{N-1} \left(\frac{\delta f_{\mathrm{U},k}[n]}{C_k}+l{_{\mathrm{U},k}^{\mathrm{off}}}[n]\right)= \sum\limits_{n{\rm{ = }}1}^{N-2} l{_{k}}[n],\ \forall k\in\mathcal{K}, \slabel{eq:WSECM1_4}\\
&&\hspace{-8mm} \sum\limits_{n{\rm{ = }}3}^N l{_{\mathrm{U},k}^{\mathrm{down}}}[n]= O_k\sum\limits_{n{\rm{ = }}2}^{N-1} \left(\frac{\delta f_{\mathrm{U},k}[n]}{C_k} +l{_{\mathrm{U},k}^{\mathrm{off}}}[n]\right), \ \forall k\in\mathcal{K}, \slabel{eq:WSECM1_5}\\
&&\hspace{-8mm}\sum\limits_{n{\rm{ = }}1}^{N}\frac{\tau}{C_k}f_{k}[n]+\sum\limits_{n{\rm{ = }}1}^{N-2} l{_{k}}[n]=I_k,\ \forall k\in\mathcal{K}, \slabel{eq:WSECM1_3}\\
&&\hspace{-8mm} B{_k^\mathrm{off}}[n]+B{_{\mathrm{U},k}^{\mathrm{off}}}[n]+B{_{\mathrm{U},k}^{\mathrm{down}}}[n]=B, \ \forall n\in\mathcal{N}, \ \forall k\in\mathcal{K}, \slabel{eq:WSECM1_12}\\
&&\hspace{-8mm} f_k[n]\geq0, \ \forall n\in\mathcal{N},\ \forall k\in\mathcal{K}, \slabel{eq:WSECM1_6}\\
&&\hspace{-8mm} l_k[N-1]=l_k[N]=0, \ l_k[n]\geq0,\ \forall n\in\mathcal{N}_1,\  \forall k\in\mathcal{K}, \slabel{eq:WSECM1_11}\\
&&\hspace{-8mm} f_{\mathrm{U},k}[1]=f_{\mathrm{U},k}[N]=0,\ f_{\mathrm{U},k}[n]\geq0, \ \forall n\in\mathcal{N}_2, \ \forall  k\in\mathcal{K}, \slabel{eq:WSECM1_14}\\ 
&&\hspace{-8mm} l{_{\mathrm{U},k}^{\mathrm{off}}}[1]=l{_{\mathrm{U},k}^{\mathrm{off}}}[N]=0, \ l{_{\mathrm{U},k}^{\mathrm{off}}}[n]\geq0, \ \forall n\in\mathcal{N}_2,\ \forall k\in\mathcal{K}, \slabel{eq:WSECM1_7}\\ 
&&\hspace{-8mm} l{_{\mathrm{U},k}^{\mathrm{down}}}[1]=l{_{\mathrm{U},k}^{\mathrm{down}}}[2]=0,\ l{_{\mathrm{U},k}^{\mathrm{down}}}[n]\geq0, \ \forall n\in\mathcal{N}_3,\ \forall k\in\mathcal{K}, \slabel{eq:WSECM1_8}\\
&&\hspace{-8mm} B{_k^\mathrm{off}}[N-1]=B{_k^\mathrm{off}}[N]=0,\ B{_k^\mathrm{off}}[n]\geq0,\ \forall n\in\mathcal{N}_1,\ \forall k\in\mathcal{K},\slabel{eq:WSECM1_13}\\
&&\hspace{-8mm} B{_{\mathrm{U},k}^{\mathrm{off}}}[1]=B{_{\mathrm{U},k}^{\mathrm{off}}}[N]=0,\ B{_{\mathrm{U},k}^{\mathrm{off}}}[n]\geq0,
\ \forall n\in\mathcal{N}_2,\ \forall k\in\mathcal{K},\slabel{eq:WSECM1_15}\\ 
&&\hspace{-8mm} B{_{\mathrm{U},k}^{\mathrm{down}}}[1]=B{_{\mathrm{U},k}^{\mathrm{down}}}[2]=0,\ B{_{\mathrm{U},k}^{\mathrm{down}}}[n]\geq0, \forall n\in\mathcal{N}_3,\ \forall k\in\mathcal{K}, \slabel{eq:WSECM1_16}\\
&&\hspace{-8mm} \mathbf{u}[0]=\mathbf{u}_\mathrm{I}, \ \mathbf{u}[N]=\mathbf{u}_\mathrm{F}, \slabel{eq:WSECM1_9}\\
&&\hspace{-8mm} \|\mathbf{u}[n]-\mathbf{u}[n-1]\|\leq V_\mathrm{max}\tau,\ \forall n\in\mathcal{N},  \quad \slabel{eq:WSECM1_10}
\end{subeqnarray}
}}
\hspace{-2mm}where $\mathbf{z}\triangleq\{\mathbf{z}_k[n]\}_{k\in\mathcal{K},n\in\mathcal{N}}$ and $\mathbf{B}\triangleq\{\mathbf{B}_k[n]\}_{k\in\mathcal{K},n\in\mathcal{N}}$ with $\mathbf{z}_k[n]\triangleq\{ f_{k}[n],l{_{k}}[n], f_{\mathrm{U},k}[n], l{_{\mathrm{U},k}^{\mathrm{off}}}[n], l{_{\mathrm{U},k}^{\mathrm{down}}}[n]\}$ and $\mathbf{B}_k[n]\triangleq\{B{_k^\mathrm{off}}[n], B{_{\mathrm{U},k}^{\mathrm{off}}}[n], B{_{\mathrm{U},k}^{\mathrm{down}}}[n] \}$, respectively, denote the sets of the computational resource scheduling variables and the bandwidth allocation variables for UE $k$ in time slot $n$, $\mathbf{u}\triangleq\{\mathbf{u}[n]\}_{n\in\mathcal{N}}$ denotes the set of the UAV's horizontal locations for all the slots, i.e., the trajectory of the UAV, and $\mathcal{N}_1=\{1,\dots,N-2\}$.
In (P1), \eqref{eq:WSECM1_0} is the objective function for minimizing the WSEC where $w_\mathrm{U}$ and $\{w_k\}_{k\in\mathcal{K}}$ represent the weights of the UAV and UEs, respectively, which trade-offs between the UAV and UEs, and the priority/fairness among the UEs. Also, \eqref{eq:WSECM1_1} and \eqref{eq:WSECM1_2} are the two information-causality constraints, while  \eqref{eq:WSECM1_4}--\eqref{eq:WSECM1_3} are the UEs' computation task constraints to make sure that all the UEs' computation task-input data have been computed and the task-output data have been received. The bandwidth constraints are in \eqref{eq:WSECM1_12}, while \eqref{eq:WSECM1_6}--\eqref{eq:WSECM1_16} ensure the non-negativeness of the optimization variables. \eqref{eq:WSECM1_9} and \eqref{eq:WSECM1_10}  specify the UAV's initial and final horizontal locations, and its maximum speed constraints.

\section{Algorithm Design}\label{algorithm_design1}
The problem (P1) is a complicated non-convex optimization problem because of the non-convex objective function where non-linear couplings exist among the variables $l{_{k}}[n]$ and $B{_k^\mathrm{off}}[n]$, $l{_{\mathrm{U},k}^{\mathrm{off}}}[n]$ and $B{_{\mathrm{U},k}^{\mathrm{off}}}[n]$, $l{_{\mathrm{U},k}^{\mathrm{down}}}[n]$ and $B{_{\mathrm{U},k}^{\mathrm{down}}}[n]$ for $k\in \mathcal{K}, n \in\mathcal{N}$, and these variables are also strongly coupled with the trajectory of the UAV, i.e., $\mathbf{u}[n]$. To address these issues, we propose a three-step alternating optimization algorithm to solve the problem. In the first step, the computation resource scheduling variables in $\mathbf{z}$ are optimized by solving the problem with  given UAV trajectory $\mathbf{u}$ and bandwidth allocation $\mathbf{B}$; and then in the second step, the bandwidth allocation variables in $\mathbf{B}$ will be optimized with the same given UAV trajectory $\mathbf{u}$ and the optimized  $\mathbf{z}$ obtained in the first step; and finally in the third step, we focus on designing the UAV trajectory $\mathbf{u}$ with the optimized variables $\mathbf{z}$ and $\mathbf{B}$. The details for the three-step algorithm are presented as follows.

\subsection{Computation Resource Scheduling with Fixed UAV Trajectory and Bandwidth Allocation}\label{computation_resource_scheduling}
A sub-problem of (P1) is the computation resource scheduling problem (P1.1), where the UAV's trajectory $\mathbf{u}$ and bandwidth allocation $\mathbf{B}$ are given as fixed. In this case, the time-dependent channels  $\{h_{\mathrm{AP}}[n]\}_{n\in\mathcal{N}}$ and $\{h_k[n]\}_{k\in\mathcal{K},n\in\mathcal{N}}$ defined in \eqref{eq:channel_AP} and \eqref{eq:channel_k} are also known. Besides, the non-linear couplings among the offloading/downloading task-input/task-output \mbox{bits} ($l{_{k}}[n], l{_{\mathrm{U},k}^{\mathrm{off}}}[n], l{_{\mathrm{U},k}^{\mathrm{down}}}[n]$) with their corresponding allocated bandwidths ($B{_k^\mathrm{off}}[n], B{_{\mathrm{U},k}^{\mathrm{off}}}[n], B{_{\mathrm{U},k}^{\mathrm{down}}}[n]$) no longer exist. The resource scheduling problem (P1.1) is convex with a convex objective function and convex constraints, which is expressed as
\begin{subeqnarray}\label{eq:RS1}
 ({\rm P1.1}):
 \underset{\mathbf{z}}{\min} &&\hspace{-4mm}
  \sum\limits_{n{\rm{ = }}1}^N\left(w_\mathrm{U}E{_{\mathrm{U}}^{(1)}}[n]+\sum\limits_{k{\rm{ = }}1}^Kw_kE_k[n]\right) \quad\quad\slabel{eq:RS1_0}\\
\mathrm{s.t.}
&&\hspace{-4mm} \eqref{eq:WSECM1_1}-\eqref{eq:WSECM1_3}, \ \eqref{eq:WSECM1_6}-\eqref{eq:WSECM1_8}, \slabel{eq:RS1_1} 
\end{subeqnarray}
where $E{_{\mathrm{U}}^{(1)}}[n]=\sum\limits_{k{\rm{ = }}1}^K \Big(E_{\mathrm{U},k}[n]+E{_{\mathrm{U},k}^{\mathrm{off}}}[n]+
E{_{\mathrm{U},k}^{\mathrm{down}}}[n]\Big)$.
In order to gain more insights of the solution, we leverage the Lagrange method  \cite{B_Boyd04Convex} to solve problem (P1.1), and the optimal solution of problem (P1.1) is given in the following theorem.

\begin{theorem}\label{theorem_1}
The optimal solution of problem (P1.1) related to UE $k\in\mathcal{K}$ is given in \eqref{eq:f_k}--\eqref{eq:l_down_Uk} at the top of the next page,
\begin{figure*}
\begin{align}
f{_{k}^*}[n]&=\sqrt{\frac{[\beta{_k^*}]^+}{3C_kw_k\kappa_k}}, \  n\in \mathcal{N}, \label{eq:f_k}\\
l{_{k}^*}[n]&=\left\{
\begin{aligned}
&\delta B{_k^{\mathrm{off}}}[n]\Bigg[\varphi_k[n]+\log_2\bigg[\sum\limits_{i{\rm{ = }}n+1}^{N-1}\lambda{_{k,i}^*}+\beta{_k^*}-\eta{_k^*} \bigg]^+\Bigg]^+, \  n\in \mathcal{N}_1, \\
&0, \quad n=N-1~\mathrm{or}~N, \\
\end{aligned}\right. \label{eq:l_k}\\
f{_{\mathrm{U},k}^*}[n]&=\left\{
\begin{aligned}
&\sqrt{\frac{ \bigg[\eta{_k^*}-O{_k}\rho{_k^*}+O_k\sum\limits_{i{\rm{ = }}{n+1}}^{N}\mu{_{k,i}^*}-\sum\limits_{i{\rm{ = }}n}^{N-1}\lambda{_{k,i}^*}
 \bigg]^+}{3C_kw_\mathrm{U}\kappa_\mathrm{U}}}, \ n\in \mathcal{N}_2, \\
&0, \quad n=1~\mathrm{or}~N, \\
\end{aligned}\right. \label{eq:f_Uk}\\
l{_{\mathrm{U},k}^{\mathrm{off}*}}[n]&=\left\{
\begin{aligned}
&\delta B{_{\mathrm{U},k}^{\mathrm{off}}}[n]\Bigg[\varphi{_{\mathrm{U},k}^{\mathrm{off}}}[n]+ \log_2\bigg[\eta{_k^*}-O{_k}\rho{_k^*}+O_k\sum\limits_{i{\rm{ = }}{n+1}}^{N}\mu{_{k,i}^*}-\sum\limits_{i{\rm{ = }}n}^{N-1}\lambda{_{k,i}^*} \bigg]^+\Bigg]^+, \ n\in \mathcal{N}_2,  \\
&0, \quad n=1~\mathrm{or}~N, \\
\end{aligned}\right. \label{eq:l_off_Uk}\\
l{_{\mathrm{U},k}^{\mathrm{down}*}}[n]&=\left\{
\begin{aligned}
&\delta B{_{\mathrm{U},k}^{\mathrm{down}}}[n]\Bigg[\varphi{_{\mathrm{U},k}^{\mathrm{down}}}[n]+\log_2\bigg[\rho{_k^*}-
\sum\limits_{i{\rm{ = }}n}^{N} \mu{_{k,i}^*} \bigg]^+\Bigg]^+, \ n\in \mathcal{N}_3,  \\
&0, \quad n=1~\mathrm{or}~2, \\
\end{aligned}\right. \label{eq:l_down_Uk}
\end{align}
\end{figure*}
where
\begin{align}
\varphi_k[n]&=\log_2\frac{B{_k^{\mathrm{off}}}[n]h_k[n]}{w_kN_0\ln2},\ n\in\mathcal{N}_1, \label{eq:phi_k} \\
\varphi{_{\mathrm{U},k}^{\mathrm{off}}}[n]&=\log_2\frac{B{_{\mathrm{U},k}
^{\mathrm{off}}}[n]h_{\mathrm{AP}}[n]}{w_\mathrm{U}N_0\ln2},\ n\in\mathcal{N}_2, \label{eq:phi_off_Uk}\\
\varphi{_{\mathrm{U},k}^{\mathrm{down}}}[n]&=\log_2\frac{B{_{\mathrm{U},k}
^{\mathrm{down}}}[n]h_{k}[n]}{w_\mathrm{U}N_0\ln2},\ n\in\mathcal{N}_3, \label{eq:phi_down_Uk}
\end{align}
are denoted as the offloading/downloading priority indicators for the UEs in each given slot. Also, $\lambda{_{k,n}^*}\geq0$ and  $\mu{_{k,n}^*}\geq0$ for $k\in\mathcal{K}, n\in\mathcal{N}$ are respectively the optimal Lagrange multipliers (dual variables) associated with the inequality constraints \eqref{eq:WSECM1_1} and \eqref{eq:WSECM1_2} in problem (P1.1) (or P1), while  $\eta_k^*$, $\rho^*_k$ and $\beta^*_k$  are respectively the optimal Lagrange multipliers  associated with the equality constraints \eqref{eq:WSECM1_4}--\eqref{eq:WSECM1_3} for $k\in\mathcal{K}$.
\end{theorem}

\begin{proof}
See Appendix A.
\end{proof}

\begin{remark}\label{rem:1}
(Intuitive Explanation). From the expressions relating to the computation resource scheduling \mbox{parameters} in Theorem \ref{theorem_1}, we observe that  $\{l{_{k}}[n]\}$, $\{l{_{\mathrm{U},k}^{\mathrm{off}}}[n]\}$, and $\{l{_{\mathrm{U},k}^{\mathrm{down}}}[n]\}$  are monotonically increasing with $\{\varphi_k[n]\}$, $\{\varphi{_{\mathrm{U},k}^{\mathrm{off}}}[n]\}$ and $\{\varphi{_{\mathrm{U},k}^{\mathrm{down}}}[n]\}$ when they are positive.
It coincides with the intuition that more input (or output) data should be offloaded (or downloaded) with larger $\{\varphi_k[n]\}$, $\{\varphi{_{\mathrm{U},k}^{\mathrm{off}}}[n]\}$ and $\{\varphi{_{\mathrm{U},k}^{\mathrm{down}}}[n]\}$, corresponding to the scenarios with larger bandwidths, channel power gains  and smaller weights for energy consumption.
\end{remark}

\begin{remark}\label{rem:2}
(Decreasing Offloading and Increasing Downloading Data Size). Theorem \ref{theorem_1} sheds light on the fact that $l{_{k}^*}[n]$ decreases with the time slot index $n$ while $l{_{\mathrm{U},k}^{\mathrm{down}*}}[n]$ increases with $n$ for the reason that $\sum_{i{\rm{ = }}n+1}^{N-1}\lambda{_{k,i}^*}$ and $\sum_{i{\rm{ = }}n}^{N} \mu{_{k,i}^*}$ in \eqref{eq:l_k} and \eqref{eq:l_down_Uk} decrease with $n$ as $\lambda{_{k,i}^*}\geq0$ and $\mu{_{k,i}^*}\geq0$. This indicates that the resource allocated for UEs' task offloading gradually decreases while that for UAV's downloading gradually increases as time goes by.
\end{remark}

It is necessary to obtain the optimal values of the Lagrange multipliers, i.e., $\boldsymbol{\lambda}^*=\{\lambda{_{k,n}^*}\}_{k\in\mathcal{K}, n\in\mathcal{N}}$, $\boldsymbol{\mu}^*=\{\mu{_{k,n}^*}\}_{k\in\mathcal{K}, n\in\mathcal{N}}$, $\boldsymbol{\eta}^*=\{\eta{_{k}^*}\}_{k\in\mathcal{K}}$, $\boldsymbol{\rho}^*=\{\rho{_{k}^*}\}_{k\in\mathcal{K}}$ and  $\boldsymbol{\beta}^*=\{\beta{_{k}^*}\}_{k\in\mathcal{K}}$  since they play important roles in determining the optimal computation resource scheduling $\mathbf{z}^*$ according to Theorem \ref{theorem_1}. In this paper, we adopt a subgradient-based algorithm to obtain the optimal dual variables in  $\boldsymbol{\lambda}^*$ and $\boldsymbol{\mu}^*$  related to the inequality constraints \eqref{eq:WSECM1_1}, \eqref{eq:WSECM1_2}, as described in the following \mbox{Lemma \ref{lemma1}}.

\begin{lemma}\label{lemma1}
The dual variables $\{\lambda_{k,n}\}$ and $\{\mu_{k,n}\}$ obtained at the $(j+1)$-th ($j=1,2,\dots$) iteration of the subgradient-based algorithm are expressed as
\begin{align}
\lambda_{k,n,j+1}&=[\lambda_{k,n,j}-\varepsilon^{(\lambda)}_j\Delta\lambda_{k,n,j}]^+, \ k\in\mathcal{K}, n\in\mathcal{N}_2, \label{eq:lambda_kn} \\
\mu_{k,n,j+1}&=[\mu_{k,n,j}-\varepsilon^{(\mu)}_j\Delta\mu_{k,n,j}]^+, \ k\in\mathcal{K}, n\in\mathcal{N}_3,  \label{eq:mu_kn}
\end{align}
with the corresponding subgradients given as
{\small{
\begin{align}
\hspace{-2mm}\Delta\lambda_{k,n,j}&=\sum\limits_{i{\rm{ = }}1}^{n-1} l{_{k,j}^*}[i]-\sum\limits_{i{\rm{ = }}2}^n \left(\frac{\delta f{_{\mathrm{U},k,j}^*}[i]}{C_k}+l{_{\mathrm{U},k,j}^{\mathrm{off}*}}[i]
\right),  \label{eq:lambda1_kn} \\
\hspace{-2mm}\Delta\mu_{k,n,j}&=O_k\sum\limits_{i{\rm{ = }}2}^{n-1} \left(\frac{\delta f{_{\mathrm{U},k,j}^*}[i]}{C_k} +l{_{\mathrm{U},k,j}^{\mathrm{off}*}}[i]\right)-\sum\limits_{i{\rm{ = }}3}^n l{_{\mathrm{U},k,j}^{\mathrm{down}*}}[i],  \label{eq:mu1_kn}
\end{align}
}}
\hspace{-1.4mm}where $\varepsilon^{(\lambda)}_j$ and $\varepsilon^{(\mu)}_j$ respectively  denote the iterative steps for obtaining the dual variables in $\boldsymbol{\lambda}$ and $\boldsymbol{\mu}$ at the $j$-th iteration \cite{B_D.Bertsekas89Parallel}.
Also, $\{l{_{k,j}^*}[n]\}$, $\{f{_{\mathrm{U},k,j}^*}[i]\}$, $\{l{_{\mathrm{U},k,j}^{\mathrm{off}*}}[i]\}$, $\{l{_{\mathrm{U},k,j}^{\mathrm{down}*}}[n]\}$  are the  computation resource scheduling variables obtained through Theorem \ref{theorem_1} with the dual variables obtained at the $j$-th iteration, i.e.,  $\boldsymbol{\lambda}_j=\{\lambda_{k,n,j}\}_{k\in\mathcal{K}, n\in\mathcal{N}}$, $\boldsymbol{\mu}_j=\{\mu_{k,n,j}\}_{k\in\mathcal{K}, n\in\mathcal{N}}$, $\boldsymbol{\eta}_j=\{\eta_{k,j}\}_{k\in\mathcal{K}}$, $\boldsymbol{\rho}_{j}=\{\rho_{k,j}\}_{k\in\mathcal{K}}$ and $\boldsymbol{\beta}_j=\{\beta_{k,j}\}_{k\in\mathcal{K}}$.
\end{lemma}

Besides, the bi-section search method is used to obtain the optimal dual variables in $\boldsymbol{\eta}^*$, $\boldsymbol{\rho}^*$ and  $\boldsymbol{\beta}^*$  related to the equality constraints \eqref{eq:WSECM1_4}--\eqref{eq:WSECM1_3}, as summarized  in Lemma \ref{lemma2}.

\begin{lemma}\label{lemma2}
With the obtained $\boldsymbol{\lambda}_{j+1}$ and $\boldsymbol{\mu}_{j+1}$ above, the corresponding $\boldsymbol{\eta}_{j+1}$, $\boldsymbol{\rho}_{j+1}$ and  $\boldsymbol{\beta}_{j+1}$ can be obtained by bi-section search of $\{\beta_{k,j+1}\}_{k\in\mathcal{K}}\in[0,\{\beta_{k,\mathrm{max}}\}_{k\in\mathcal{K}})$ where $\beta_{k,\mathrm{max}}=3C_kw_k\kappa_k(\frac{I_kC_k}{T})^2$. For each given $\beta_{k,j+1}\in[0,\beta_{k,\mathrm{max}})$, the corresponding $\eta_{k,j+1}$ and $\rho_{k,j+1}$ can be obtained with another two bi-section searches within $\eta_{k,j+1}\in[\eta_{k,j+1}^{\mathrm{low}},\eta_{k,j+1}^{\mathrm{up}}]$ and $\rho_{k,j+1}\in[\rho_{k,j+1}^{\mathrm{low}},\rho_{k,j+1}^{\mathrm{up}}]$ to make the expressions satisfy \eqref{eq:C_7}$=$\eqref{eq:C_8} and \eqref{eq:C_7}$=$\eqref{eq:C_9},  respectively, in Appendix B, where the expressions of $\eta_{k,j+1}^{\mathrm{low}}$, $\eta_{k,j+1}^{\mathrm{up}}$, $\rho_{k,j+1}^{\mathrm{low}}$, and $\rho_{k,j+1}^{\mathrm{up}}$ are given in \eqref{eq:C_11}--\eqref{eq:C_14} in Appendix B.
The optimal $\beta_{k,j+1}$, $\eta_{k,j+1}$ and $\rho_{k,j+1}$ should satisfy \eqref{eq:C_7}$=$\eqref{eq:C_10}.
\end{lemma}

\begin{proof}
See Appendix B.
\end{proof}

The optimal dual variables $\boldsymbol{\lambda}^*,\boldsymbol{\mu}^*$ and $\boldsymbol{\eta}^*,\boldsymbol{\rho}^*, \boldsymbol{\beta}^*$ can be finally obtained when the subgradient algorithm converges, and the bi-section searches terminate. Note that  the corresponding convergence can be guaranteed according to \cite{B_Boyd04Convex}.

\subsection{Bandwidth Allocation with Fixed UAV Trajectory and Computation Resource Scheduling}\label{bandwidth_allocation}
Here, another sub-problem of (P1), denoted as the bandwidth allocation problem (P1.2) is considered to optimize $\mathbf{B}$ with the same given UAV's trajectory $\mathbf{u}$  and the optimized computation resource scheduling parameters in $\mathbf{z}$.
The bandwidth allocation problem (P1.2) is expressed as
\begin{subeqnarray}\label{eq:BA}
 ({\rm P1.2}):
 \underset{\mathbf{B}}{\min} &&\hspace{-4mm}
  \sum\limits_{n{\rm{ = }}1}^N\left(w_\mathrm{U}E{_{\mathrm{U}}^{(2)}}[n]+\sum\limits_{k{\rm{ = }}1}^Kw_kE{_k^{\mathrm{off}}}[n]\right) \quad\quad \slabel{eq:BA_0}\\
\mathrm{s.t.}
&&\hspace{-4mm} \eqref{eq:WSECM1_12}, \eqref{eq:WSECM1_13}-\eqref{eq:WSECM1_16}, \slabel{eq:BA_1} 
\end{subeqnarray}
where $E{_{\mathrm{U}}^{(2)}}[n]=\sum\limits_{k{\rm{ = }}1}^K \Big(E{_{\mathrm{U},k}^{\mathrm{off}}}[n]+
E{_{\mathrm{U},k}^{\mathrm{down}}}[n]\Big)$. It can be easily proved that problem (P1.2) is convex with  convex objective function and constraints. To gain more insights on the structure of the optimal solution, we again leverage the Lagrange method \cite{B_Boyd04Convex} to solve this problem, and the optimal solution to problem (P1.2) is given in the following theorem.

\begin{theorem}\label{theorem_2}
The optimal solution of problem (P1.2)  related to UE $k\in\mathcal{K}$ is given by
\vspace{-3mm}

{\small{
\begin{align}
\hspace{-2mm}B{_k^{\mathrm{off}*}}[n]&=\left\{
\begin{aligned}
&\frac{\frac{\ln2}{2}l_k[n]}{\delta W_0\big[\frac{\ln2}{2}(\frac{\phi_{k,n}}{w_k}h_k[n]l_k[n])^{\frac{1}{2}}\big]},~n\in\mathcal{N}_1, \\
&0, \quad n=N-1~\mathrm{or}~N, \\
\end{aligned}\right. \label{B_k_off}\\
\hspace{-2mm}B{_{\mathrm{U},k}^{\mathrm{off}*}}[n]&=\left\{
\begin{aligned}
&\frac{\frac{\ln2}{2}l{_{\mathrm{U},k}^{\mathrm{off}}}[n]}{\delta W_0\big[\frac{\ln2}{2}(\frac{\phi{_{k,n}}}{w_\mathrm{U}}h_{\mathrm{AP}}[n]l{_{\mathrm{U},k}^{\mathrm{off}}}[n])^{\frac{1}{2}}\big]},~n\in\mathcal{N}_2, \\
&0, \quad n=1~\mathrm{or}~N, \\
\end{aligned}\right. \label{B_Uk_off}\\
\hspace{-2mm}B{_{\mathrm{U},k}^{\mathrm{down}*}}[n]&=\left\{
\begin{aligned}
&\frac{\frac{\ln2}{2}l{_{\mathrm{U},k}^{\mathrm{down}}}[n]}{\delta W_0\big[\frac{\ln2}{2}(\frac{\phi{_{k,n}}}{w_\mathrm{U}}h_k[n]l{_{\mathrm{U},k}^{\mathrm{down}}}[n])^{\frac{1}{2}}\big]},~n\in\mathcal{N}_3, \\
&0, \quad n=1~\mathrm{or}~2, \\
\end{aligned}\right. \label{B_Uk_down}
\end{align}
}}
\hspace{-1mm}where $\phi_{k,n}=\frac{\nu{_{k,n}^*}}{\delta^2N_0\ln2}$ with $\{\nu{_{k,n}^*}\}_{k\in\mathcal{K}, n\in\mathcal{N}}$  being the optimal Lagrange multipliers (dual variables) associated with the equality constraints in \eqref{eq:WSECM1_12} of problem (P1.2) (or P1), and  $W_0(x)$ is the principal branch of the Lambert $W$ function defined as the solution of $W_0(x)e^{W_0(x)}=x$ \cite{J_Corless1996LambertW}.
\end{theorem}

\begin{proof}
See Appendix C.
\end{proof}

\begin{lemma}\label{lemma3}
(Exclusive Bandwidth Allocation). According to the optimal bandwidth allocation results in Theorem \ref{theorem_2} combining with the equality constraints in \eqref{eq:WSECM1_12}, we have
\begin{align}
B{_k^{\mathrm{off}*}}[n]&=B, \ \mathrm{if}~ l_k[n]>0, l{_{\mathrm{U},k}^{\mathrm{off}}}[n]=l{_{\mathrm{U},k}^{\mathrm{down}}}[n]=0, \label{eq:BB_UE} \\
B{_{\mathrm{U},k}^{\mathrm{off}*}}[n]&=B, \ \mathrm{if}~  l{_{\mathrm{U},k}^{\mathrm{off}}}[n]>0, l_k[n]=l{_{\mathrm{U},k}^{\mathrm{down}}}[n]=0,  \label{eq:BB_U_off}\\
B{_{\mathrm{U},k}^{\mathrm{down}*}}[n]&=B, \ \mathrm{if}~ l{_{\mathrm{U},k}^{\mathrm{down}}}[n]>0, l_k[n]=l{_{\mathrm{U},k}^{\mathrm{off}}}[n]=0,  \label{eq:BB_U_down}
\end{align}
where the whole bandwidth is exclusively occupied when only one of $l_k[n]$, $l{_{\mathrm{U},k}^{\mathrm{off}}}[n]$, $l{_{\mathrm{U},k}^{\mathrm{down}}}[n]$ is positive for any $k\in\mathcal{K}$, $n\in\mathcal{N}$.
Also, it is always sure that
\begin{align} \label{B1BN}
B{_k^{\mathrm{off}*}}[1]=B, \ B{_{\mathrm{U},k}^{\mathrm{down}*}}[N]=B, \ k\in\mathcal{K}.
\end{align}
\end{lemma}

The optimal Lagrange multipliers $\{\nu{_{k,n}^*}\}$ for obtaining the optimal bandwidth allocation in Theorem \ref{theorem_2} correspond to $\{\phi_{k,n}\}$, which should make the equality constraints in \eqref{eq:WSECM1_12} satisfied. In fact, $\phi_{k,n}$ can be obtained effectively with the bi-section search when the bandwidth is not exclusively occupied, i.e., at least two of $l_k[n]$, $l{_{\mathrm{U},k}^{\mathrm{off}}}[n]$, $l{_{\mathrm{U},k}^{\mathrm{down}}}[n]$ are positive,  since $\{B{_k^{\mathrm{off}*}}[n]\}_{n\in\mathcal{N}_1}$, $\{B{_{\mathrm{U},k}^{\mathrm{off}*}}[n]\}_{n\in\mathcal{N}_2}$ and $\{B{_{\mathrm{U},k}^{\mathrm{down}*}}[n]\}_{n\in\mathcal{N}_3}$ are all monotonically decreasing functions with respect to (w.r.t.) $\{\phi_{k,n}\}$ according to the property of the $W_0$ function. Besides, we can obtain tight search ranges using the results in Lemma~\ref{lemma4}.

\begin{lemma}\label{lemma4}
A tight bi-section search range of $\phi_{k,n}$ ($k\in\mathcal{K}$) for any slot $n\in\mathcal{N}$ with non-exclusive bandwidth is given as $\phi_{k,n}\in[\phi_{k,n}^\mathrm{min},\phi_{k,n}^\mathrm{max}]$ where
\begin{align} \label{phi_range}
&&&\phi_{k,n}^\mathrm{min}~(\mathrm{or}~\phi_{k,n}^\mathrm{max})=\min~(\mathrm{or}~\max)\\
&&&\left\{
\begin{aligned}
&\{\phi_{\mathrm{UE},k,n}(B/3),
\phi_{\mathrm{U},k,n}^{\mathrm{off}}(B/3),\phi_{\mathrm{U},k,n}^{\mathrm{down}}(B/3)\},~\mathrm{\mathbf{case~1}}\\
&\{\phi_{\mathrm{UE},k,n}(B/2),\phi_{\mathrm{U},k,n}^{\mathrm{off}}(B/2)\}, ~~\mathrm{\mathbf{case~2}}\\
&\{\phi_{\mathrm{UE},k,n}(B/2),
\phi_{\mathrm{U},k,n}^{\mathrm{down}}(B/2)\}, ~~\mathrm{\mathbf{case~3}}\\
&\{\phi_{\mathrm{U},k,n}^{\mathrm{off}}(B/2),
\phi_{\mathrm{U},k,n}^{\mathrm{down}}(B/2)\}, ~~~\mathrm{\mathbf{case~4}}
\end{aligned}\right. \nonumber
\end{align}
where \emph{\textbf{case 1}-\textbf{case 4}} are distinguished by the values of $l_k[n]$, $l{_{\mathrm{U},k}^{\mathrm{off}}}[n]$ and $l{_{\mathrm{U},k}^{\mathrm{down}}}[n]$ for each $n\in\mathcal{N}$. For \emph{\textbf{case 1}}, all the three parameters have positive values; for \emph{\textbf{case 2}}, $l{_{\mathrm{U},k}^{\mathrm{down}}}[n]=0$; for \emph{\textbf{case 3}}, $l{_{\mathrm{U},k}^{\mathrm{off}}}[n]=0$; for \emph{\textbf{case 4}}, $l_k[n]=0$. In \eqref{phi_range},
\begin{align}
\phi_{\mathrm{UE},k,n}(x)&=\frac{w_kl_k[n]}{\delta^2 x^2h_k[n]}e^{\frac{l_k[n]\ln2}{\delta x}},\ k\in\mathcal{K}, n\in\mathcal{N}, \\
\phi_{\mathrm{U},k,n}^{\mathrm{off}}(x)&=\frac{w_\mathrm{U}l{_{\mathrm{U},k}^{\mathrm{off}}}[n]}{\delta^2 x^2 h_{\mathrm{AP}}[n]}e^{\frac{l{_{\mathrm{U},k}^{\mathrm{off}}}[n]\ln2}{\delta x}},\ k\in\mathcal{K}, n\in\mathcal{N}, \\
\phi_{\mathrm{U},k,n}^{\mathrm{down}}(x)&=\frac{w_\mathrm{U}l{_{\mathrm{U},k}^{\mathrm{down}}}[n]}{\delta^2 x^2 h_{k}[n]}e^{\frac{l{_{\mathrm{U},k}^{\mathrm{down}}}[n]\ln2}{\delta x}}, \ k\in\mathcal{K}, n\in\mathcal{N},
\end{align}
which are the value of $\phi_{k,n}$ obtained by letting the expressions of $B{_k^{\mathrm{off}*}}[n]$, $B{_{\mathrm{U},k}^{\mathrm{off}*}}[n]$ and $B{_{\mathrm{U},k}^{\mathrm{down}*}}[n]$ in \eqref{B_k_off}--\eqref{B_Uk_down} equal to $x$.
\end{lemma}

\subsection{UAV Trajectory Design With Fixed Computation Resource Scheduling and Bandwidth Allocation}\label{UAV_trajectory_design}
Here, the sub-problem for designing the UAV's trajectory $\mathbf{u}$ is considered, which we refer to it as the UAV trajectory design problem (P1.3), by assuming that the computation resource scheduling  $\mathbf{z}$  and bandwidth allocation $\mathbf{B}$ are given as fixed with the previously optimized values. Hence, the  UAV trajectory design problem (P1.3) can be rewritten as
\begin{subeqnarray}\label{eq:UAV}
 ({\rm P1.3}):
 \underset{\mathbf{u}}{\min} &&\hspace{-4mm}
  \sum\limits_{n{\rm{ = }}1}^N\left(w_\mathrm{U}E{_{\mathrm{U}}^{(3)}}[n]+\sum\limits_{k{\rm{ = }}1}^Kw_kE{_k^{\mathrm{off}}}[n]\right)\quad\quad \slabel{eq:UAV_0}\\
\mathrm{s.t.}
&&\hspace{-4mm} \eqref{eq:WSECM1_9}, \eqref{eq:WSECM1_10}, \slabel{eq:UAV_1}
\end{subeqnarray}
where $E{_{\mathrm{U}}^{(3)}}[n]=E{_\mathrm{U}^{\mathrm{fly}}}[n]+\sum\limits_{k{\rm{ = }}1}^K \Big(E{_{\mathrm{U},k}^{\mathrm{off}}}[n]+E{_{\mathrm{U},k}^{\mathrm{down}}}[n]\Big)$.
It is noted that the $E{_\mathrm{U}^{\mathrm{fly}}}[n]$ defined in \eqref{eq:E_fly} with $v[n]$ in \eqref{eq:speed_n} is not a convex function of $\mathbf{u}$. In order to address this issue, we define an upper bound of $E{_\mathrm{U}^{\mathrm{fly}}}[n]$ as follows
\begin{align}\label{eq:E_fly_up}
\widetilde{E}{_\mathrm{U}^{\mathrm{fly}}}[n]=\tau\bigg(\theta_1v^3[n]+\frac{\theta_2}{\widetilde{v}[n]}\bigg), \ n\in \mathcal{N},
\end{align}
by introducing a variable $\widetilde{v}[n]$ and a constraint $v[n]\geq \widetilde{v}[n]$, which is equivalent to $\|\mathbf{u}[n]-\mathbf{u}[n-1]\|^2\geq \widetilde{v}^2[n]\tau^2$.
This constraint  is still non-convex, and we leverage the SCA technique to solve this issue.
The left hand side of the constraint is convex versus $\mathbf{u}$ and can be  approximated as its linear lower bound by using the first-order Taylor expansion at a local point $\mathbf{u}_i$, where $i=1,2,\dots$ denotes the iteration index of the SCA method. Hence, the additional constraint can be approximated as a convex one as follows
\begin{align}\label{eq:speed_n_app}
&\widetilde{v}^2[n]\tau^2-2(\mathbf{u}_i[n]-\mathbf{u}_i[n-1])^T(\mathbf{u}[n]-\mathbf{u}[n-1]) \\ \nonumber
&\leq \|\mathbf{u}_i[n]-\mathbf{u}_i[n-1]\|^2, \ n\in \mathcal{N}.
\end{align}
The approximated problem of (P1.3) with $\{\widetilde{E}{_\mathrm{U}^{\mathrm{fly}}}[n]\}$, $\{\widetilde{v}[n]\}$ and the additional constraint \eqref{eq:speed_n_app} is convex w.r.t.~$\mathbf{u}$ and $\{\widetilde{v}[n]\}$. However, the UAV's locations in different slots are coupled with each other as in \eqref{eq:WSECM1_10}, and thus it is difficult to obtain a closed-form solution of $\mathbf{u}$. In this case, we resort to the software CVX \cite{M_Grant08CVX} to solve the approximated problem of (P1.3).

\subsection{Algorithm, Convergence and Complexity}
Based on the aforementioned analysis of the alternating optimization for the computation resource scheduling $\mathbf{z}$, the bandwidth allocation $\mathbf{B}$ and the UAV trajectory $\mathbf{u}$ in each subproblem, Algorithm~\ref{algorithmic1} is proposed to solve the original problem (P1) for obtaining the solution $\{\mathbf{z}^*,\mathbf{B}^*,\mathbf{u}^*\}$.\footnote{The proposed method is not theoretically optimal due to problem non-convexity, but its performance gain is verified by the simulation results.} 

\begin{algorithm}[!htp]
\caption{Three-Step Algorithm for Solving Problem (P1)  }\label{algorithmic1}
\begin{algorithmic}[1]
\STATE \textbf{Set} $B$, $T$, $N$, $K$,  $h_0$, $N_0$, $H$, $V_{\mathrm{max}}$, $\theta_1$, $\theta_2$, $\mathbf{u}_\mathrm{I}$,  $\mathbf{u}_\mathrm{F}$, $w_\mathrm{U}$, $\kappa_\mathrm{U}$, $\{\mathbf{v}_k, w_k, I_k, C_k, O_k, \kappa_k\}_{k\in\mathcal{K}}$,  two tolerant thresholds $\epsilon_1$ and $\epsilon$, and the iterative steps  $\{\varepsilon^{(\lambda)}_j\}$ and $\{\varepsilon^{(\mu)}_j\}$;
\STATE \textbf{Initialize} the iteration index $\zeta=1$ and $\mathbf{u}_{1}$, $\mathbf{B}_{1}$;
\STATE \textbf{Repeat 1}
\STATE~~\textbf{Initialize}  $j=1$, as well as $\boldsymbol{\lambda}_1$, $\boldsymbol{\mu}_1$;
\STATE~~\textbf{Step 1: Repeat 1.1}
\STATE
      ~\quad a) Obtain  $\boldsymbol{\eta}_j$, $\boldsymbol{\rho}_{j}$, $\boldsymbol{\beta}_j$ with $\boldsymbol{\lambda}_j$, $\boldsymbol{\mu}_j$ through \textbf{Lemma }\ref{lemma2}; \\
      ~\quad b) Obtain  $\mathbf{z}_{\zeta,j}^*=\big\{\{f{_{k,j}^*}[n]\}, \{l{_{k,j}^*}[n]\}, \{f{_{\mathrm{U},k,j}^*}[n]\},$\\
      ~~~~~~~$\{l{_{\mathrm{U},k,j}^{\mathrm{off}*}}[n]\}, \{l{_{\mathrm{U},k,j}^{\mathrm{down}*}}[n]\}\big\}$ through \textbf{Theorem} \ref{theorem_1} with \\
      ~~~~~~~$\boldsymbol{\lambda}_j$, $\boldsymbol{\mu}_j$, $\boldsymbol{\eta}_j$, $\boldsymbol{\rho}_{j}$, $\boldsymbol{\beta}_j$ and $\mathbf{u}_{\zeta}$, $\mathbf{B}_{\zeta}$; \\
      ~\quad c) Calculate the WSEC $E{_j^{(1)}}$ by substituting $\mathbf{z}_{\zeta,j}^*$, $\mathbf{B}_{\zeta}$,\\
      ~~~~~~~$\mathbf{u}_{\zeta}$ into the objective function of problem (P1.1); \\
      ~\quad d) $j=j+1$; \\
      ~\quad e) Update $\boldsymbol{\lambda}_j$ and $\boldsymbol{\mu}_j$ according to \textbf{Lemma }\ref{lemma1}; \\
\STATE~~\textbf{End Repeat 1.1} until convergence, i.e., $|E{_j^{(1)}}-E{_{j-1}^{(1)}}|<$ \\
      ~~$\epsilon_1$ ($j>1$), and obtain optimal $\mathbf{z}_{\zeta+1}=\mathbf{z}_{\zeta,j}^*$;\\
\STATE~~\textbf{Step 2: Bi-section search} of $\{\nu_{k,n}\}$ to find the optimal \\
 ~~$\{\nu_{k,n}^*\}$ and obtain the $\mathbf{B}_{\zeta+1}=\mathbf{B}_{\zeta}^*=\big\{\{B{_{k}^{\mathrm{off}*}}[n]\},$ \\ ~~$\{B{_{\mathrm{U},k}^{\mathrm{off}*}}[n]\},\{B{_{\mathrm{U},k}^{\mathrm{down}*}}[n]\}\big\}$ according to \textbf{Theorem} \ref{theorem_2}, \\
 ~~\textbf{Lemma} \ref{lemma3}  and \textbf{Lemma} \ref{lemma4} with given $\mathbf{u}_{\zeta}$ and $\mathbf{z}_{\zeta+1}$;
\STATE ~~\textbf{Step 3:} Solve the approximated problem of (P1.3)  by\\
  ~~CVX based on the SCA method, so as to obtain the \\
  ~~optimal solution $\mathbf{u}_{\zeta+1}$ with the given $\mathbf{z}_{\zeta+1}$, $\mathbf{B}_{\zeta+1}$; \\
\STATE $\zeta=\zeta+1$; \\
\STATE Calculate the WSEC $E{_\zeta}$, by substituting
      $\mathbf{z}_{\zeta}$, $\mathbf{B}_{\zeta}$, and $\mathbf{u}_{\zeta}$ into the objective function of problem (P1); \\
\STATE\textbf{End Repeat 1} until convergence, i.e., $|E{_\zeta}-E{_{\zeta-1}}|<\epsilon$ ($\zeta>2$), and obtain the minimum WSEC $E{_\zeta}$ with the solution $\mathbf{z}^*=\mathbf{z}_{\zeta}$, $\mathbf{B}^*=\mathbf{B}_{\zeta}$, $\mathbf{u}^*=\mathbf{u}_{\zeta}$; \\
\end{algorithmic}
\end{algorithm}

The convergence of Algorithm \ref{algorithmic1} is easy to prove in light of the guaranteed convergence of the loop Repeat 1.1 in Step 1, the bi-section search in Step 2 and the CVX solving process based on the SCA method in Step 3 \cite{B_Boyd04Convex}.
The lower-bounded objective function of problem (P1) will monotonically decrease with the iteration index $\zeta$ by  optimizing $\mathbf{z}$, $\mathbf{B}$ and $\mathbf{u}$ alternatingly in each sub-problem, which further guarantees the convergence of the algorithm.

\setlength{\tabcolsep}{0.3 pt}\begin{table*}[thb]
\centering
\caption{Simulation Parameters}\label{table1}
\begin{tabular}{|l|l|l|}
\hline
~\textbf{Parameter }&~{\textbf{Symbol}} &~{\textbf{Value}} \\
\hline
~The total system bandwidth\quad\quad\quad\quad\quad\quad\quad\quad\quad\quad\quad\quad\quad\quad~   &~$B$ \quad\quad\quad\quad\quad\quad\quad &~30 MHz \quad\quad\quad\quad\quad\quad\quad\quad\quad\quad\quad\quad \\
\hline
~The total  task completion time  &~$T$ &~10 seconds\\
\hline
~Number of time slots &~$N$ &~50  \\
\hline
~Number of ground UEs &~$K$ &~4  \\
\hline
~The channel power gain at a reference distance of $d_0$=1 m &~$h_0$ &~$-30\mathrm{dB}$ \\
\hline
~The noise power  &~$N_0$ &~$-60$dBm \\
\hline
~The fixed altitude of the UAV&~$H$ &~10 m \\
\hline
~The maximum available speed of the UAV&~$V_{\mathrm{max}}$ &~10 m/s \\
\hline
~The UAV's propulsion energy consumption related parameters&~($\theta_1,\theta_2$) &~(0.00614,15.976)\\
\hline
~The initial and final position of the UAV  &~$\mathbf{u}_\mathrm{I}$, $\mathbf{u}_\mathrm{F}$ &~$(-5,-5)$, $(5,-5)$\\
\hline
~The horizontal positions of the UEs  &~$\mathbf{v}_1$, $\mathbf{v}_2$, $\mathbf{v}_3$, $\mathbf{v}_4$  &~$(5,5)$, $(-5,5)$, $(-5,-5)$, $(5,-5)$\\
\hline
~The effective switched capacitance of the UAV and UEs &~$\kappa_\mathrm{U}$, $\kappa_k (k\in\mathcal{K})$ &~$10^{-28}$ \\
\hline
~The weight for energy consumption of the UAV &~$w_\mathrm{U}$ &~$0.2$ \\
\hline
~The weight for energy consumption of the UEs&~$w_k~(k\in\mathcal{K})$ &~$1$ \\
\hline
~Required CPU cycles per bit &~$C_k~(k\in\mathcal{K})$ &~1000 cycles/bit \\
\hline
~UEs' task-input data size &~$I_k~(k\in\mathcal{K})$ &~400 Mbits \\
\hline
~UEs' task size ratio of output data to input data&~$O_k~(k\in\mathcal{K})$ &~$0.8$ \\
\hline
~The tolerant thresholds &~$\epsilon_1$ and $\epsilon$ &~$10^{-4}$  \\
\hline
\end{tabular}
\end{table*}

In addition, Algorithm~\ref{algorithmic1} is easy to implement and the corresponding  complexity is acceptable. In Step 1, the complexity mainly comes from the subgradient method for obtaining $\{\lambda_{k,n}\}$, $\{\mu_{k,n}\}$,   and the bi-section searches of $\{\beta_{k}\}$, $\{\rho_{k}\}$ and $\{\eta_{k}\}$ in each iteration of Repeat 1.1. Let $\varepsilon_{\mathrm{sub}}>0$, and $\varepsilon_{\beta}, \varepsilon_{\rho}, \varepsilon_{\eta}>0$ denote the computational accuracies of the subgradient method and the bisectional searches for  $\{\beta_{k}\}$, $\{\rho_{k}\}$ and $\{\eta_{k}\}$. Thus, the corresponding complexity can be calculated as $\mathcal{O}(1/\varepsilon_{\mathrm{sub}}^2+K\log_2(1/\varepsilon_{\beta})(\log_2(1/\varepsilon_{\rho})+\log_2(1/\varepsilon_{\eta})) )$. In Step 2, the complexity is from the bisection search of $\{\nu_{k,n}\}$, which is calculated as $\mathcal{O}(KN\log_2(1/\varepsilon_{\nu}))$, where $\varepsilon_{\nu}$ is the corresponding computational accuracy.
In Step 3, the complexity mainly focuses on solving the approximation problem of (P1.3) by CVX, which is acceptable in general.

\section{Simulation Results}\label{sec:simulation}

In this section, simulation results are presented to evaluate the performance of the proposed algorithm against the benchmarking schemes. The effects of the key parameters will be analyzed, including the relative location of the AP ($\mathbf{v}_0$),\footnote{In order to properly show the effects of the relative location of the AP to UEs on UAV's trajectory and the performance, we fix the locations of the UEs and vary the location  of AP even though AP is usually fixed in practice.} the computation task sizes of UEs ($I_k$ for $k\in\mathcal{K}$), the task completion time for UEs ($T$), the size ratio of task-output data to task-input data ($O_k$ for $k\in\mathcal{K}$), the weight for energy consumption of the UAV ($w_\mathrm{U}$), and the iteration index of the alternating optimization algorithm ($\zeta$).
The basic simulation parameters are listed in Table~\ref{table1} unless specified otherwise.

\subsection{Trajectory of the UAV}\label{sec:Trajectory}
In this subsection, numerical results for the trajectory of the UAV are given to shed light on the effects of  the task sizes of UEs
($[I_1,I_2,I_3,I_4]$) and the relative location of the AP ($\mathbf{v}_0$).
In Fig. \ref{Trajectory}, the UAV's flying trajectories are depicted in different scenarios.
It should be noted that the total task size of UEs is same for the cases in (a), (c), (d) and (f), i.e., 1400 Mbits, while the cases for (b) and (e) are with larger total task size, e.g., 1800 Mbits. From these results in Fig. \ref{Trajectory}, we can observe that the trajectory of the UAV is heavily reliant on the relative location of the AP and the distribution of UEs' task sizes.

\begin{figure}[tbp]
\centering
\includegraphics[scale=0.45]{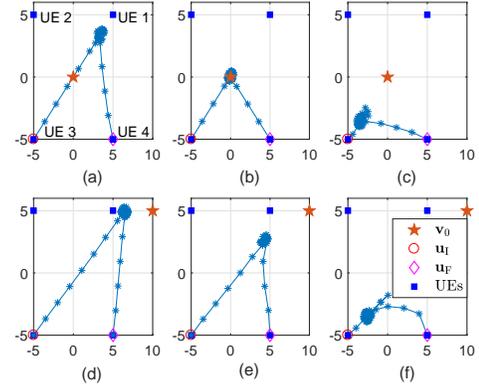}
\caption{The trajectories of the UAV in the situations with \mbox{different} location of the AP and task size allocation of the UEs:   $\mathbf{v}_0=(0,0)$ for (a), (b) and (c), $\mathbf{v}_0=(10,5)$ for (d), (e) and (f);  $[I_1,I_2,I_3,I_4]=[6,2,4,2]\times10^2$Mbits for (a) and (d), $[I_1,I_2,I_3,I_4]=[6,4,6,2]\times10^2$Mbits for (b) and (e), $[I_1,I_2,I_3,I_4]=[2,2,6,4]\times10^2$Mbits for (c) and (f).}
\label{Trajectory}
\end{figure}

For the scenario of $\mathbf{v}_0=(0,0)$, the AP is surrounded by the UEs and at the center of the UEs' distributed area. We can observe that the UAV tends to fly close to the UEs with large task sizes and tries to be not too far away from the AP when the total task sizes of UEs are moderate as the results in cases (a) and (c). When the total task size becomes larger and the distribution of UEs' task sizes becomes  more average, the UAV tends to fly close to the AP as the result in case (b).
These three cases  indicate that for the scenario where the AP is located at the center of UEs' distributed area, the distribution of the UEs' task sizes plays an important role on the UAV's trajectory, while the effect of the AP's location will become more dominant when the UEs' total task size becomes larger, which coincides with the intuition that more task-input data will be offloaded to the AP in this situation so as to reduce the WSEC by making use of the super computing resources at the AP. For the scenario of $\mathbf{v}_0=(10,5)$, the AP is located outside the distributed area of the UEs and its average distance to the UEs is relatively larger than the above scenario.
In this situation, the effects of AP's location on the trajectories are more prominent, where the comparison between (a) and (d), (b) and (e), (c) and (f) can properly explain this.

The reason behind these results in Fig.~\ref{Trajectory} is that there exists a tradeoff between the distribution of UEs' task sizes and the relative location of the AP to the UEs. In other words, getting close to the UEs with large task sizes can reduce UEs' offloading and UAV's downloading energy consumption, while being closer to the AP will reduce the UAV's offloading energy consumption, and thus the UAV has to find a balance between these two factors meanwhile taking its own flying energy consumption into consideration, so as to minimize the WSEC through optimizing its flying trajectory.

\subsection{Performance Improvement}\label{sec:Performance}
Here, we focus on the performance gain of the proposed algorithm. The performance of the baselines is also provided for comparison, including
the ``Direct Trajectory" scheme where the UAV flies from its initial location to the final location directly with an average speed;
the ``Offloading Only" scheme where the UEs just rely on task offloading to the UAV and the AP for computing without local computing by the UEs themselves; the ``Equal Bandwidth" scheme indicating the solution that the whole bandwidth are equally divided  by the active $B{_k^{\mathrm{off}}}[n]$, $ B{_{\mathrm{U},k}^{\mathrm{off}}}[n]$, and $B{_{\mathrm{U},k}^{\mathrm{down}}}[n]$, for $n\in\mathcal{N}$ and $k\in\mathcal{K}$ without bandwidth optimization;
and the ``Local Computing" scheme, where the UEs rely on their own computing resources to complete their computation tasks without offloading. Note that the former four schemes are all offloading schemes. To better illustrate the effects of AP's relative location on the performance, we present all the results in two scenarios given in Fig.~\ref{Trajectory}, i.e., $\mathbf{v}_0=(0,0)$ and $\mathbf{v}_0=(10,5)$.

\begin{figure}[tbp]
\centering
\includegraphics[scale=0.45]{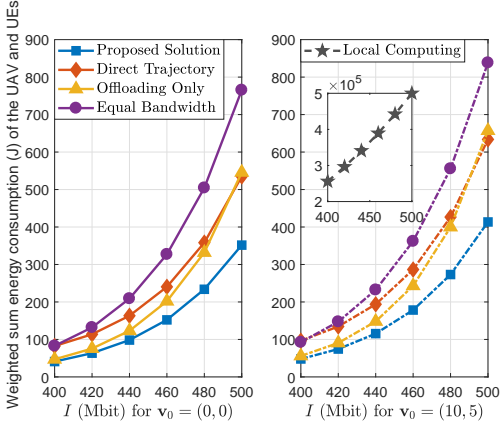}
\caption{The WSEC of the UAV and UEs versus the uniform task size: $I=I_k$ for $k\in\mathcal{K}$.}
\label{Energy_I}
\end{figure}

Fig.~\ref{Energy_I} shows the WSEC results versus the uniform task size $I=I_k$ for $k\in\mathcal{K}$. All the curves in the figures increase with $I$ as expected since more energy will be consumed by completing tasks with more input data. It can be seen that great performance improvement can be achieved by leveraging the proposed solution in comparison with all the baseline  schemes in both scenarios. It is clear that the performance of the ``Local Computing" scheme is far worse that the other schemes with computation offloading, verifying the importance of edge computing through offloading.
Specifically, the WSECs of the ``Proposed Solution" are almost one thousandth of that for the ``Local Computing" scheme, presenting the tremendous benefits the UEs obtained by deploying the UAV as an assistant for computing and  relaying.
In addition,  the WSECs of the proposed solution are half less than those of the ``Equal Bandwidth" scheme and they are almost quarter less than those of the ``Direct Trajectory" scheme. The ``Offloading Only" scheme performs well with relatively small task sizes, e.g., $I=400$ Mbits, but its gaps between the ``Proposed Solution" are even larger than those of the ``Direct Trajectory" scheme when task sizes are large, e.g., $I=500$ Mbits. All these results verify that the proposed optimization on bandwidth allocation and UAV's trajectory, as well as making full use of the computing resources at UEs  have great effect on  minimizing the WSEC of the UAV and UEs.
Note that the gaps between the proposed solution and the baselines become larger when $I$ increases, which further indicates that the proposed algorithm is more capable of handling the computation-intensive tasks.

\begin{figure}[tbp]
\centering
\includegraphics[scale=0.45]{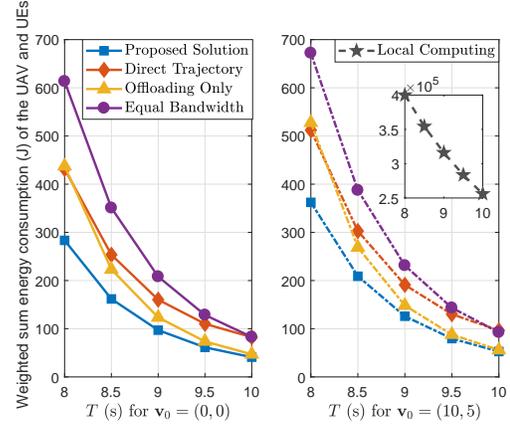}
\caption{The WSEC  of the UAV and UEs versus the total task completion time: $T$ (s).}
\label{Energy_T}
\end{figure}

In Fig.~\ref{Energy_T}, the WSEC  w.r.t.~the total task completion time $T$ is depicted.
We can see that the WSECs of all the schemes decrease with $T$, coinciding with the intuition that a tradeoff exists between the energy consumption and time consumption for completing the same tasks, and the energy consumption will decrease when the consumed time increases.
It is notable that the proposed solution is superior than the four baseline schemes in both scenarios, and the performance improvement is even more prominent with strict time restriction (small $T$), which further confirms that the proposed algorithm is good at dealing with the latency-critical computation tasks and can achieve a better energy-delay tradeoff. Besides, some similar insights can also be obtained as from Fig.~\ref{Energy_I}.

\begin{figure}[tbp]
\centering
\includegraphics[scale=0.45]{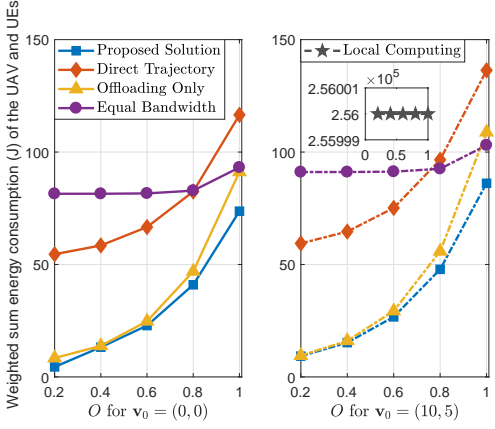}
\caption{The WSEC of the UAV and UEs versus the uniform size ratio of task-output data to task-input data: $O=O_k$ for $k\in\mathcal{K}$.}
\label{Energy_Ok}
\end{figure}

Fig.~\ref{Energy_Ok} depicts the WSEC w.r.t.~the uniform size ratio of the task-output data to the task-input data $O=O_k$ for $k\in\mathcal{K}$.
We see that the proposed scheme outperforms the baselines in both scenarios as in Fig.~\ref{Energy_I} and Fig.~\ref{Energy_T}.
The WSEC of the ``Local Computing" scheme is constant w.r.t $O$, while
the WSECs of all the other schemes increase with $O$ since more output data will be downloaded to the UEs in the cases with larger $O$. However, the curves of the ``Equal Bandwidth" scheme are almost unchanged for $O\in[0.2,0.8]$ due to the fact that  equally allocated bandwidth to the downloading transmissions should be sufficient to complete the downloading missions,
and its performance is much worse than the other offloading schemes for smaller $O$ because of the irrational bandwidth allocation.
Note that the gaps between the proposed solution and the ``Direct Trajectory" scheme decrease as $O$ increases since it becomes more difficult to balance the tradeoff between UEs' task sizes and the relative location of the AP.
In comparison, the gaps between the proposed solution and the ``Offloading Only" scheme become large as $O$ increases for the  reason that local computing may be an energy-saving way when with a large $O$. In the scenario of $\mathbf{v}_0=(10,5)$, the ``Offloading Only" scheme performs even worse than the ``Equal Bandwidth" scheme when $O=1$, which further verifies that the effect of partial local computing in minimizing the WSEC.

\begin{figure}[tbp]
\centering
\includegraphics[scale=0.45]{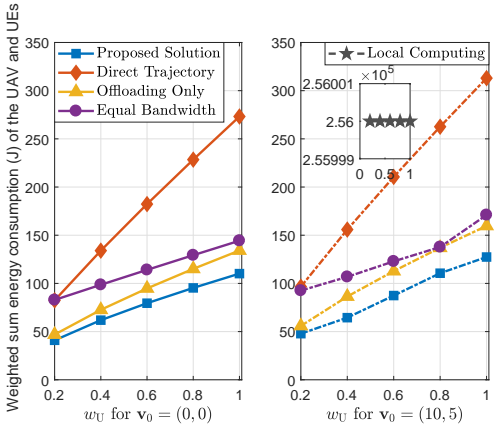}
\caption{The WSEC of the UAV and UEs versus the weight for energy consumption of the UAV: $w_\mathrm{U}$.}
\label{Energy_wU}
\end{figure}

\begin{figure}[tbp]
\centering
\includegraphics[scale=0.45]{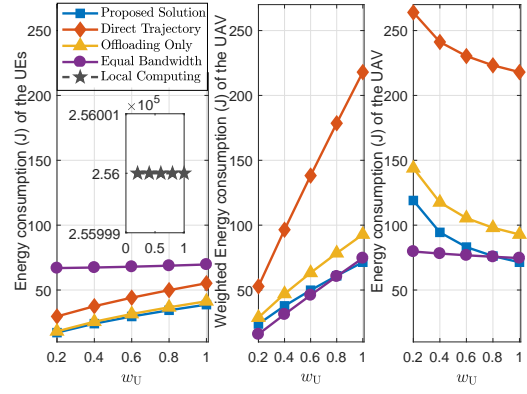}
\caption{Separate energy consumption of the UEs and the UAV versus the weight for energy consumption of the UAV: $w_\mathrm{U}$.}
\label{3Energy_wU}
\end{figure}

Results for the WSEC versus the UAV's weight  $w_\mathrm{U}$ are shown in Fig.~\ref{Energy_wU}. It is clear that the proposed scheme still performs best in both scenarios.  All the curves increase with $w_\mathrm{U}$ except that for ``Local Computing" scheme, since larger proportion of  UAV's energy consumption will be calculated into the WSEC with a larger $w_\mathrm{U}$.
Note that the gaps between the proposed solution and the ``Direct Trajectory"  scheme become obviously larger as $w_\mathrm{U}$ increases in both scenarios especially compared with those gaps related to the ``Offloading Only" and the ``Equal Bandwidth" schemes.
This is due to the fact that the energy consumption for UAV's propulsion contributes a larger part for WSEC of the ``Direct Trajectory" scheme without trajectory optimization, and thus its WSEC  increases  much faster  w.r.t.~$w_\mathrm{U}$ than the other schemes.

From the above results, we can observe that the WSEC for the scenario of $\mathbf{v}_0=(10,5)$ is higher than that for the scenario of $\mathbf{v}_0=(0,0)$ for all the schemes. It is easy to understand that more energy will be used for UAV's offloading transmission and flying because of the farther average distance between the AP and UEs. The performance of the proposed scheme is also more stable than that of the baseline schemes  considering the changing of the relative location of the AP to UEs since its relative WSEC increment is the smallest among the schemes.

Based on Fig.~\ref{Energy_wU}, we depict the energy consumption of the UEs (also the weighted energy consumption of the UEs with $w_1=w_2=w_3=w_4=1$), the weighted energy consumption and the energy consumption of the UAV versus $w_\mathrm{U}$ in Fig.~\ref{3Energy_wU} (a), (b) and (c), respectively. It is clear that the weighted energy consumption of the UEs and the UAV for the four offloading schemes increase with $w_\mathrm{U}$ as in (a) and (b), while their energy consumption of the UAV decreases  with $w_\mathrm{U}$ as in (c).
This is due to the fact that we aim at minimizing the WSEC, and the objectives increase with $w_\mathrm{U}$ similar to the results in Fig.~\ref{Energy_wU}. Meanwhile minimizing the UAV's energy consumption becomes more important as $w_\mathrm{U}$ increases.
From this figure, we can better see the tremendous benefits obtained by the UEs from the UAV, especially when $w_\mathrm{U}$ is smaller. In the case of $w_\mathrm{U}=0.2$, the UAV consumes 120 Joule of energy to help the UEs decrease their energy consumption from $2.56*10^5$ Joule of the ``Local Computing" scheme to 20 Joule of the ``Proposed Solution", by providing assistance of task computing and relaying (further offloading to the AP for computing) through the proposed algorithm.

\begin{figure}[tbp]
\centering
\includegraphics[scale=0.44]{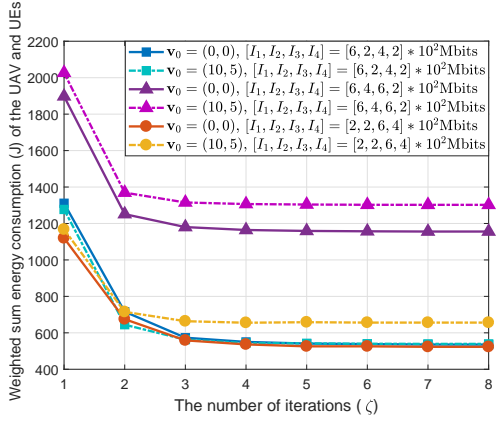}
\caption{The WSEC of the UAV and UEs versus the number of iteration: $\zeta$.}
\label{Energy_zeta}
\end{figure}

Fig. \ref{Energy_zeta} shows the WSEC of the proposed solution w.r.t to the iteration index $\zeta$ under different settings. From the figure, we can see that the proposed solution almost converges at $\zeta=3$, i.e., after twice iteration of optimizing $\mathbf{z}$, $\mathbf{B}$ and $\mathbf{u}$, regardless of the UEs' task sizes or the position of the AP.

\section{Conclusion}\label{sec:conclusion}
This paper investigated the UAV-assisted MEC architecture, where the UAV acts as an MEC server and a relay to assist the UEs to compute their tasks or further offload their tasks to the AP for computing. We minimized the WSEC of the UAV and the UEs under some practical constraints, using an alternating algorithm iteratively optimizing the computation resource scheduling, bandwidth allocation, and the UAV's trajectory. The simulation results have confirmed that the UAV's trajectory is greatly affected by the relative location of the AP and the distribution of UEs' task sizes. Besides, significant performance improvement and more stable performance can be achieved by  the proposed algorithm over the baseline schemes.


\section*{Appendix A: Proof of {Theorem}~\ref{theorem_1}}
\label{App:theo_1}
\renewcommand{\theequation}{A.\arabic{equation}}
\setcounter{equation}{0}

The partial Lagrange function of (P1.1) can be expressed as
\vspace{-5mm}

{\small{
\begin{align}\label{eq:L_P1_1}
 &\mathcal{L}^{(1)}(\mathbf{z},\boldsymbol{\lambda},\boldsymbol{\mu},\boldsymbol{\eta},\boldsymbol{\rho},\boldsymbol{\beta})=
  \sum\limits_{k{\rm{ = }}1}^K \Bigg\{\sum\limits_{n{\rm{ = }}1}^N \bigg(w_k\Big(E{_k^{\mathrm{local}}}[n]+E{_k^{\mathrm{off}}}[n]\Big) \nonumber \\
 & +w_\mathrm{U}\Big(E_{\mathrm{U},k}[n]+E{_{\mathrm{U},k}^{\mathrm{off}}}[n]+E{_{\mathrm{U},k}^{\mathrm{down}}}[n]\Big)\bigg) \nonumber \\
 &  +\Bigg(\sum\limits_{n{\rm{ = }}2}^{N-1} \widetilde{\lambda}_{k,n} \bigg(\frac{\delta f_{\mathrm{U},k}[n]}{C_k}+l{_{\mathrm{U},k}^{\mathrm{off}}}[n]\bigg)- \sum\limits_{n{\rm{ = }}1}^{N-2} \widehat{\lambda}_{k,n}  l{_{k}}[n]\Bigg) \nonumber \\
 &  +\Bigg( \sum\limits_{n{\rm{ = }}3}^N \widetilde{\mu}_{k,n} l{_{\mathrm{U},k}^{\mathrm{down}}}[n]- O_k \sum\limits_{n{\rm{ = }}2}^{N-1}  \widehat{\mu}_{k,n}  \left(\frac{\delta f_{\mathrm{U},k}[n]}{C_k} +l{_{\mathrm{U},k}^{\mathrm{off}}}[n]\right)\Bigg) \nonumber \\
 &  +\eta_{k}\Bigg(\sum\limits_{n{\rm{ = }}1}^{N-2} l{_{k}}[n]-
  \sum\limits_{n{\rm{ = }}2}^{N-1} \bigg(\frac{\delta f_{\mathrm{U},k}[n]}{C_k}+l{_{\mathrm{U},k}^{\mathrm{off}}}[n]\bigg)\Bigg) \nonumber \\
 &  +\rho_{k} \Bigg(O_k\sum\limits_{n{\rm{ = }}2}^{N-1} \bigg(\frac{\delta f_{\mathrm{U},k}[n]}{C_k} +l{_{\mathrm{U},k}^{\mathrm{off}}}[n]\bigg)-\sum\limits_{n{\rm{ = }}3}^N l{_{\mathrm{U},k}^{\mathrm{down}}}[n]\Bigg) \nonumber \\
 & + \beta_k  \Bigg(I_k -\sum\limits_{n{\rm{ = }}1}^{N-2} l{_{k}}[n] -\sum\limits_{n{\rm{ = }}1}^{N}\frac{\tau}{C_k}f_{k}[n]\Bigg)\Bigg\},
\end{align}
}}
\hspace{-1.5mm}where  $\boldsymbol{\lambda}=\{\lambda_{k,n}\}_{k\in\mathcal{K}, n\in\mathcal{N}}$, $\boldsymbol{\mu}=\{\mu_{k,n}\}_{k\in\mathcal{K}, n\in\mathcal{N}}$, $\boldsymbol{\eta}=\{\eta_k\}_{k\in\mathcal{K}}$, $\boldsymbol{\rho}=\{\rho_k\}_{k\in\mathcal{K}}$, $\boldsymbol{\beta}=\{\beta_k\}_{k\in\mathcal{K}}$, $\widetilde{\lambda}_{k,n}=\sum_{i{\rm{ = }}n}^{N-1} \lambda_{k,i}$, $\widehat{\lambda}_{k,n}=\sum_{i{\rm{ = }}{n+1}}^{N-1} \lambda_{k,i}$, $\widetilde{\mu}_{k,n}=\sum_{i{\rm{ = }}n}^N \mu_{k,i}$, and $\widehat{\mu}_{k,n}=\sum_{i{\rm{ = }}{n+1}}^N \mu_{k,i}$.
The Lagrangian dual function of problem (P1.1) can be presented as
\begin{align}\label{Dual_fun_P1_1}
d^{(1)}(\boldsymbol{\lambda},\boldsymbol{\mu},\boldsymbol{\eta},\boldsymbol{\rho},\boldsymbol{\beta})=&~\underset{\mathbf{z}}{\min}~ ~\mathcal{L}^{(1)}(\mathbf{z},\boldsymbol{\lambda},\boldsymbol{\mu},\boldsymbol{\eta},\boldsymbol{\rho},\boldsymbol{\beta})\\
&~~\mathrm{s.t.} ~~ \eqref{eq:WSECM1_6}-\eqref{eq:WSECM1_8}.    \nonumber
\end{align}
Hence, the solution of $\mathbf{z}$ with given dual variables $\boldsymbol{\lambda},\boldsymbol{\mu},\boldsymbol{\eta},\boldsymbol{\rho},\boldsymbol{\beta}$ can be obtained by solving problem \eqref{Dual_fun_P1_1}. If the given dual variables are optimal, denoted as $\boldsymbol{\lambda}^*,\boldsymbol{\mu}^*,\boldsymbol{\eta}^*,\boldsymbol{\rho}^*,\boldsymbol{\beta}^*$, then the corresponding solutions are optimal, i.e., $\mathbf{z}^*$.
According to the structures of $\mathcal{L}^{(1)}(\mathbf{z},\boldsymbol{\lambda},\boldsymbol{\mu},\boldsymbol{\eta},\boldsymbol{\rho},\boldsymbol{\beta})$
and the constraints \eqref{eq:WSECM1_6}-\eqref{eq:WSECM1_8}, it is noted that the problem \eqref {Dual_fun_P1_1} can be equivalently divided into $K$ subproblems w.r.t.~each UE $k\in\mathcal{K}$ to facilitate parallel execution. Apply the Karush-Kuhn-Tucker (KKT) conditions \cite{B_Boyd04Convex} and let the derivations of $\mathcal{L}^{(1)}(\mathbf{z},\boldsymbol{\lambda},\boldsymbol{\mu},\boldsymbol{\eta},\boldsymbol{\rho},\boldsymbol{\beta})$ w.r.t.  $f{_{k}}[n], l{_{k}}[n], f{_{\mathrm{U},k}}[n], l{_{\mathrm{U},k}^{\mathrm{off}}}[n], l{_{\mathrm{U},k}^{\mathrm{down}}}[n]$ equal to zero, we can thus obtain the corresponding optimal solution given in Theorem \ref{theorem_1} with some straightforward calculations.

\section*{Appendix B: Proof of {Lemma}~\ref{lemma2}}
\label{App:lemma_2}
\renewcommand{\theequation}{B.\arabic{equation}}
\setcounter{equation}{0}

With the achieved $\boldsymbol{\lambda}_{j+1}$ and $\boldsymbol{\mu}_{j+1}$ in Lemma \ref{lemma1}, we can then obtain the $\boldsymbol{\eta}_{j+1}$, $\boldsymbol{\rho}_{j+1}$ and $\boldsymbol{\beta}_{j+1}$  correspondingly. According to the expressions of the optimal solution in Theorem \ref{theorem_1} and  the equality constraints in \eqref{eq:WSECM1_4}--\eqref{eq:WSECM1_3},
we can express the value of $\sum_{n{\rm{ = }}1}^{N-2} l{_{k,j+1}^*}[n]$ in the following forms in \eqref{eq:C_7}--\eqref{eq:C_10}
\vspace{-1mm}

{\small{
\begin{align}
&\hspace{-4mm}\quad\sum\limits_{n{\rm{ = }}1}^{N-2} l{_{k,j+1}^*}[n]
=I_k-\frac{T}{C_k}\sqrt{\frac{\beta{_{k,j+1}}}{3C_kw_k\kappa_k}} \label{eq:C_7}\\
&\hspace{-4mm}=\delta\sum\limits_{n{\rm{ = }}1}^{N-2} B{_k^{\mathrm{off}}}[n]\bigg[\varphi_k[n]+\log_2\Big[\widehat{\lambda}_{k,n,j+1}+\beta{_{k,j+1}}-\eta{_{k,j+1}} \Big]^+ \bigg]^+ \label{eq:C_8}\\
\hspace{-4mm}=&\frac{\delta}{O_k}\sum\limits_{n{\rm{ = }}3}^N B{_{\mathrm{U},k}^{\mathrm{down}}}[n]\bigg[\varphi{_{\mathrm{U},k}^{\mathrm{down}}}[n]+\log_2\Big[\rho{_{k,j+1}}-
\widetilde{\mu}_{k,n,j+1} \Big]^+ \bigg]^+  \label{eq:C_9}\\
&\hspace{-4mm}=\sum\limits_{n{\rm{ = }}2}^{N-1} \Bigg\{\frac{\delta}{C_k}
\sqrt{ \frac{ [\eta{_{k,j+1}}-O{_k}\rho{_{k,j+1}}+O_k\widehat{\mu}_{k,n,j+1}-\widetilde{\lambda}_{k,n,j+1}]^+ } {3C_kw_\mathrm{U}\kappa_\mathrm{U}} }  \nonumber\\
&\hspace{-4mm}\quad\quad+\delta B{_{\mathrm{U},k}^{\mathrm{off}}}[n] \bigg[\varphi{_{\mathrm{U},k}^{\mathrm{off}}}[n]+\log_2\Big[\eta{_{k,j+1}}-O{_k}\rho{_{k,j+1}} \nonumber\\
&\hspace{-4mm}\quad\quad+O_k\widehat{\mu}_{k,n,j+1}-\widetilde{\lambda}_{k,n,j+1} \Big]^+ \bigg]^+ \Bigg\}, \label{eq:C_10}
\end{align}
}}
\hspace{-1.5mm}where $\widetilde{\lambda}_{k,n,j+1}$, $\widehat{\lambda}_{k,n,j+1}$, $\widetilde{\mu}_{k,n,j+1}$, and $\widehat{\mu}_{k,n,j+1}$ are defined similar to  $\widetilde{\lambda}_{k,n}$, $\widehat{\lambda}_{k,n}$, $\widetilde{\mu}_{k,n}$, and $\widehat{\mu}_{k,n}$ in Appendix A. The expression \eqref{eq:C_7} is obtained from \eqref{eq:WSECM1_3},
\eqref{eq:C_8} comes from the expression of $\{l{_{k,j+1}^*}[n]\}$, \eqref{eq:C_9} is derived from \eqref{eq:WSECM1_4} and \eqref{eq:WSECM1_5}
with equation $\sum_{n{\rm{ = }}1}^{N-2} l{_{k,j+1}^*}[n]=\frac{1}{O_k}\sum_{n{\rm{ = }}3}^N l{_{\mathrm{U},k,j+1}^{\mathrm{down}*}}[n]$, and \eqref{eq:C_10} is obtained from \eqref{eq:WSECM1_4}.

According to \eqref{eq:C_7} and the facts that $\sum_{n{\rm{ = }}1}^{N-2} l{_{k,j+1}}[n]\in[0,I_k]$, $f{_{k}^*}[n]\geq0$, we can derive the range of $\beta_{k,j+1}\in[0,\beta_{k,\mathrm{max}})$ with $\beta_{k,\mathrm{max}}=3C_kw_k\kappa_k(\frac{I_kC_k}{T})^2$ for $k\in\mathcal{K}$. It is observed from \eqref{eq:C_7}--\eqref{eq:C_9} that $\eta_{k,j+1}$ and $\rho_{k,j+1}$ are respectively monotonic non-decreasing and non-increasing implicit functions of $\beta_{k,j+1}$, which further shows that \eqref{eq:C_10} is also a monotonic non-decreasing function of $\beta_{k,j+1}$. Hence, with the obtained $\boldsymbol{\lambda}_{j+1}$ and $\boldsymbol{\mu}_{j+1}$, and a given $\beta_{k,j+1}\in[0,\beta_{k,\mathrm{max}})$, we can derive the corresponding $\eta_{k,j+1}$ and $\rho_{k,j+1}$  from the equations constituted by \eqref{eq:C_7} in company with \eqref{eq:C_8} and \eqref{eq:C_9}, respectively, also using the bi-section search method with the ranges of $\eta_{k,j+1}\in[\eta_{k,j+1}^{\mathrm{low}},\eta_{k,j+1}^{\mathrm{up}}]$ and $\rho_{k,j+1}\in[\rho_{k,j+1}^{\mathrm{low}},\rho_{k,j+1}^{\mathrm{up}}]$,
where
\begin{align}
&\eta_{k,j+1}^{\mathrm{low}}=\widehat{\lambda}_{k,N-2,j+1}-2^{\frac{I_k/\delta-\sum_{n=1}^{N-2}B{_k^{\mathrm{off}}}[n]\varphi_k[n]}
{\sum_{n=1}^{N-2}B{_k^{\mathrm{off}}}[n]}}, \label{eq:C_11}\\
&\eta_{k,j+1}^{\mathrm{up}}=\widehat{\lambda}_{k,1,j+1}+\beta_{k,\mathrm{max}}, \label{eq:C_12}\\
&\rho_{k,j+1}^{\mathrm{low}}=\widetilde{\mu}_{k,N,j+1}, \label{eq:C_13}\\
&\rho_{k,j+1}^{\mathrm{up}}=\widetilde{\mu}_{k,3,j+1}+2^{\frac{I_kO_k/\delta-\sum_{n=3}^{N}B{_{\mathrm{U},k}^{\mathrm{down}}}[n]\varphi{_{\mathrm{U},k}^{\mathrm{down}}}[n]}
{\sum_{n=3}^{N}B{_{\mathrm{U},k}^{\mathrm{down}}}[n]}}, \label{eq:C_14}
\end{align}
which are obtained from \eqref{eq:C_8} and \eqref{eq:C_9} in combination with the definitions of  $\widehat{\lambda}_{k,n,j+1}$ and $\widetilde{\mu}_{k,n,j+1}$, and the range of $\beta_{k,j+1}$. The optimal $\beta_{k,j+1}$ and the corresponding $\eta_{k,j+1}$, $\rho_{k,j+1}$ should make the equation formed  by \eqref{eq:C_7} and \eqref{eq:C_10} satisfied, which indicates the termination of the bi-section search of $\beta_{k,j+1}$, $k\in\mathcal{K}$.

\section*{Appendix C: Proof of {Theorem}~\ref{theorem_2}}
\label{App:theo_2}
\renewcommand{\theequation}{C.\arabic{equation}}
\setcounter{equation}{0}

The partial Lagrange function of (P1.2) is defined as
\begin{align}\label{eq:L_P1_2}
  &\mathcal{L}^{(2)}(\mathbf{B},\boldsymbol{\nu})= \nonumber\\
  &\sum\limits_{k{\rm{ = }}1}^K \sum\limits_{n{\rm{ = }}1}^N \bigg(w_kE{_k^{\mathrm{off}}}[n]+
  w_\mathrm{U}\Big(E{_{\mathrm{U},k}^{\mathrm{off}}}[n]+E{_{\mathrm{U},k}^{\mathrm{down}}}[n]\Big)\bigg)+ \nonumber \\
  &\sum\limits_{k{\rm{ = }}1}^K\sum\limits_{n{\rm{ = }}1}^N\nu_{k,n}\big(B-B{_k^{\mathrm{off}}}[n]-B{_{\mathrm{U},k}^{\mathrm{off}}}[n]-B{_{\mathrm{U},k}^{\mathrm{down}}}[n]\big),
\end{align}
where $\boldsymbol{\nu}=\{\nu_{k,n}\}_{k\in\mathcal{K},n\in\mathcal{N}}$.
The Lagrangian dual function of problem (P1.2) can be presented as
\begin{align}\label{Dual_fun_P1_2}
d^{(2)}(\boldsymbol{\nu})=~&\underset{\mathbf{B}}{\min}~ \mathcal{L}^{(2)}(\mathbf{B},\boldsymbol{\nu})\\
&~\mathrm{s.t.} ~~ \eqref{eq:WSECM1_13}-\eqref{eq:WSECM1_16}.    \nonumber
\end{align}Hence, the optimal solution of $\mathbf{B}$ with optimal dual variables $\boldsymbol{\nu}^*$ can be obtained by solving  \eqref{Dual_fun_P1_2}. This problem can also be equivalently divided into $K$ subproblems w.r.t.~each UE $k\in\mathcal{K}$
to facilitate parallel execution. It is easy to note that the expressions of $E{_k^{\mathrm{off}}}[n]$, $E{_{\mathrm{U},k}^{\mathrm{off}}}[n]$  and $E{_{\mathrm{U},k}^{\mathrm{down}}}[n]$  have similar structures w.r.t.~$B{_k^{\mathrm{off}}}[n]$, $B{_{\mathrm{U},k}^{\mathrm{off}}}[n]$ and $B{_{\mathrm{U},k}^{\mathrm{down}}}[n]$, and thus the optimal solution of $B{_k^{\mathrm{off}}}[n]$, $B{_{\mathrm{U},k}^{\mathrm{off}}}[n]$ and $B{_{\mathrm{U},k}^{\mathrm{down}}}[n]$ should have similar structures according to problem \eqref{Dual_fun_P1_2}.
Next, we will take $B{_k^{\mathrm{off}}}[n]$  as an example to obtain its closed-form optimal solution versus $\nu_{k,n}^*$ for $k\in\mathcal{K}, n\in\mathcal{N}$.
Applying the KKT conditions  \cite{B_Boyd04Convex} leads to the following necessary and sufficient condition of  $B{_k^{\mathrm{off}*}}[n]$:
\begin{align}\label{L_B_k_off}
\hspace{-2mm}\frac{\partial\mathcal{L}^{(2)}(\mathbf{B},\boldsymbol{\nu})}{\partial B{_k^{\mathrm{off}*}}[n]}=\nu_{k,n}^*-\frac{l_k[n]w_kN_0\ln2}{(B{_k^{\mathrm{off}*}}[n])^2h_k[n]}
2^{\frac{l_k[n]}{B{_k^{\mathrm{off}*}}[n]\delta}}=0,
\end{align}
where the optimal dual variable $\nu_{k,n}^*$ should make sure that the equality constraint  $B{_k^{\mathrm{off}*}}[n]+B{_{\mathrm{U},k}^{\mathrm{off}*}}[n]+B{_{\mathrm{U},k}^{\mathrm{down}*}}[n]=B$ is satisfied.
It is not easy to obtain the closed-form solution of $B{_k^{\mathrm{off}*}}[n]$ through \eqref{L_B_k_off} directly.
By defining $\xi=\frac{l_k[n]}{B{_k^{\mathrm{off}*}}[n]\delta}$, the equation in \eqref{L_B_k_off} can be re-expressed as
\begin{align}\label{Equa_B_k_off}
\xi^22^\xi=\frac{\nu_{k,n}^*h_k[n]l_k[n]}{\delta^2w_kN_0\ln2}\triangleq\Gamma.
\end{align}
By applying the natural logarithm at the both sides of \eqref{Equa_B_k_off} leads to
\begin{align}\label{ln_B_k_off}
\ln\xi+\frac{\ln2}{2}\xi=\ln\Gamma^{\frac{1}{2}}.
\end{align}
Then applying the exponential operation at both sides of \eqref{ln_B_k_off}, we can obtain that
\begin{align}\label{e_B_k_off}
\frac{\ln2}{2}\xi e^{\frac{\ln2}{2}\xi}=\frac{\ln2}{2}\Gamma^{\frac{1}{2}},
\end{align}
where $e$ is the base of the natural logarithm. According to the definition and property of Lambert function \cite{J_Corless1996LambertW}, we have $\frac{\ln2}{2}\xi=W_0(\frac{\ln2}{2}\Gamma^{\frac{1}{2}})$, and finally we can express $B{_k^{\mathrm{off}*}}[n]$ as
\begin{align}\label{express_B_k_off}
B{_k^{\mathrm{off}*}}[n]=\frac{\frac{\ln2}{2}l{_{k}}[n]}{\delta W_0\big[\frac{\ln2}{2}(\frac{\phi{_{k,n}}}{w_k}h_{k}[n]l{_{k}}[n])^{\frac{1}{2}}\big]},~n\in\mathcal{N}_1.
\end{align}
Integrating with the cases $B{_k^{\mathrm{off}*}}[N-1]=B{_k^{\mathrm{off}*}}[N]=0$, the complete solution of $B{_k^{\mathrm{off}^*}}[n]$ in \eqref{B_k_off} can be obtained. The solution of $B{_{\mathrm{U},k}^{\mathrm{off}*}}[n]$ and $B{_{\mathrm{U},k}^{\mathrm{down}*}}[n]$ in \eqref{B_Uk_off} and \eqref{B_Uk_down} can be obtained in a similar way \cite{J_F.Zhou18Computation}.


\bibliographystyle{IEEEtran}
\bibliography{UAV_MEC_Relay}

\end{document}